\documentclass[preprint,3p]{elsarticle1}

\usepackage{graphicx}

\usepackage{amssymb}

\usepackage{lineno}

\usepackage{amsmath}
\usepackage{amsthm}

\usepackage{mathrsfs}

\usepackage{psfrag}
\usepackage{color}

\newtheorem{theorem}{Theorem}
\newtheorem{lemma}[theorem]{Lemma}

\newtheorem{definition}[theorem]{Definition}

\newtheorem{corollary}[theorem]{Corollary}
\newtheorem{remark}{Remark}

\newtheorem{condition}{Condition}

\newtheoremstyle{algstyle}%
  {10mm}       
  {10mm}       
  {\tt}   
  {0pt}        
  {\bfseries}  
  {\newline}   
  {10mm}       
  {\thmname{#1}\thmnumber{ #2}\thmnote{ (#3)}}          
\theoremstyle{algstyle}
\newtheorem{algorithm}{Algorithm}

\newtheoremstyle{algdashstyle}%
  {10mm}       
  {10mm}       
  {\tt}   
  {0pt}        
  {\bfseries}  
  {\newline}   
  {10mm}       
  {\thmname{#1}\thmnumber{ #2}$'$\thmnote{ (#3)}}          
\theoremstyle{algdashstyle}

\newcommand{\qedsymbl}{\rule{2mm}{2mm}}

\newcommand{\nw}[1]{%
\textbf{#1}%
}

\newcommand{\mnw}[1]{%
\boldsymbol{#1}%
}

\newcommand{\bbmatrix}[1]{%
\begin{bmatrix} #1 \end{bmatrix}%
}

\newcommand{\ppmatrix}[1]{%
\begin{pmatrix} #1 \end{pmatrix}%
}

\newcommand{\mydot}[1]{%
\stackrel{\text{\Large .}}{#1}%
}

\newcommand{\lrar}{\leftrightarrow}
\newcommand{\V}{\mbox{$\cal V$}} 

\newcommand{\F}{\mbox{$\cal F$}} 

\newcommand{\dsw}{{\dot{w}}}            
\newcommand{\Vuy}{{\cal V}_{UY}}            

    \newcommand{\0}{{\mathbf 0}}        
        
 \newcommand{\Vab}{{\cal V}_{AB}}   			

\newcommand{\Vabc}{{\cal V}_{ABC}} 
\newcommand{\Vs}{{\cal V}_{S}}             
\newcommand{\Vsp}{{\cal V}_{SP}}           			
\newcommand{\Vsq}{{\cal V}_{SQ}}    

\newcommand{\Vp}{{\cal V}_{P}}              
              
\newcommand{\Vpq}{{\cal V}_{PQ}}            			
\newcommand{\Vwdwuy}{{\cal V}_{W\mydot{W}UY}}    
\newcommand{\Vwdw}{{\cal V}_{W\mydot{W}}}  			
\newcommand{\wdw}{{W\mydot{W}}}  			
\newcommand{\dw}{{\mydot{W}}}  			

\newcommand{\KSP}{\mbox{${\cal K}_{SP}$}}    		
\newcommand{\KSQ}{\mbox{${\cal K}_{SQ}$}}    		
\newcommand{\KPQ}{\mbox{${\cal K}_{PQ}$}}    		

\newcommand{\f}{\mbox{${\bf f}$}}        				
\newcommand{\fS}{\mbox{${\bf f}_{S}$}}  				
\newcommand{\fP}{\mbox{${\bf f}_{P}$}}  				
\newcommand{\fQ}{\mbox{${\bf f}_{Q}$}}  				
\newcommand{\x}{\mbox{${\bf x}$}}             		

\begin{document}

\begin{frontmatter}



\title{On the linear static output feedback problem: the annihilating polynomial approach}

\author[hn]{H. Narayanan \corref{cor1}}
\ead{hn@ee.iitb.ac.in}
\cortext[cor1]{Corresponding author}
\author[hari]{Hariharan Narayanan}
\ead{hariharan.narayanan@tifr.res.in}
\address[hn]{Department of Electrical Engineering, Indian Institute of Technology Bombay}
\address[hari]{School of Technology and Computer Science, Tata Institute of Fundamental Research}

\begin{abstract}
One of the fundamental open problems in control theory 
is that of the stabilization of a linear time invariant dynamical system
through static output feedback. We are given a linear dynamical system defined through
\begin{align*}
 \mydot{w} &= Aw + Bu \\
 y &= Cw .
\end{align*}
The problem is to find, if it exists, a feedback $u=Ky$ such that 
the matrix $A+BKC$ has all its eigenvalues in the complex left half plane
and, if such a feedback does not exist, to prove that it does not.
Substantial progress has not been made on the computational aspect 
of the solution
to this problem.

In this paper we consider instead `which annihilating polynomials can a matrix
of the form $A+BKC$ possess?'. 

We give a simple solution to this problem 
when the system has either a single input or a single output.
For the multi input - multi output case, 
we use these ideas
to 
characterize the annihilating polynomials
when $K$ has rank one, and
suggest possible computational solutions for general $K.$
We also present some numerical evidence 
for the plausibility of this approach for the general case as well
as for the problem of shifting the eigenvalues of the system.

\end{abstract}


\end{frontmatter}


\section{Introduction}
The linear static output feedback problem has been open for more than
five decades. Its importance has been stressed in many references \cite{bel},\cite{brocket},\cite{kalman},\cite{skogestad},\cite{sontag}.
While substantial progress has been made in theoretical 
understanding of this problem, from a computational point of view,
the situation is still unsatisfactory.
In this paper we suggest a technique 
for altering through output feedback, the coefficient vector of the annihilating polynomial of the state transition matrix. 
This technique is a simple modification of  old ideas
of state feedback for shifting poles \cite{Wonham1967},\cite{Wonham1978}. 
It shows promise also as a way of altering the roots of the annihilating
polynomials to more desirable locations.

\subsection{The output feedback problem}
Let
\begin{align}
\label{eqn:state111}
 \mydot{w} &= Aw + Bu 
\end{align}
\begin{align}
\label{eqn:state222}
 y &= Cw .
\end{align}
be a set of equations where $A,B,C$ respectively are $ n\times n,n\times m$ and $p\times n$
real matrices.

{\it Find, if it exists, a real matrix $K,$ such that the matrix $A+BKC$
has a specified set of eigenvalues, and if such a $K$ does not exist, prove that it does not.}

 Note that we would get $\mydot{w}=(A+BKC)w,$ if we set $u=Ky,$ 
which is called linear static output feedback.

In this paper, we consider instead, the following problem.

{\it Find, if it exists, a real matrix $K,$ such that the matrix $A+BKC$
has a specified annihilating polynomial, and if such a $K$ does not exist, prove that it does
not.} 

We solve this problem for the case where there is either a single 
input or a single output. For the multi-input multi-output case, we suggest
a variation which involves repeated application of the solution to the 
single input case.

Most of the literature on the output feedback problem, quite naturally, is devoted
to altering roots rather than coefficient vector of the annihilating polynomial.
Detailed surveys are available in \cite{sadabadi}, \cite{syrmos}.
An insightful historical summary is given in \cite{bel}.

Some of the papers which exemplify the different approaches are 
\begin{itemize}
\item the proof that the case $n<m+p$ is generically solvable
\cite{kimura},
necessary and sufficient conditions for solvability of this case
when the eigenvalues specified are distinct \cite{fletcher};
\item use of exterior algebra to solve the case $mp>n$ \cite{brocket}
\cite{ravi},\cite{wang1},\cite{wang2};
\item use of linear matrix inequalities and Lyapunov functions method \cite{cao},\cite{rosinova},\cite{scherer},\cite{skogestad},\cite{sontag},\cite{vesely}
.
\end{itemize}

In \cite{bel}
a new approach to the choice of an initial approximation for iterative procedures resolving the feedback design problem, is presented. The method appears useful when
$m+p\ \leq \ n.$

We now give a brief sketch of the ideas in this paper.

The main object which we use for modifying matrices is a
`full Krylov sequence'. A sequence $\omega \equiv (w^0, \cdots , w^n)$ of $n\times 1$ vectors is said to be be a full Krylov sequence for an $n\times n$ matrix $A,$
iff $w^j=A^jw^0, j= 1, \cdots n$ and $w^0, \cdots , w^{n-1}$
are linearly independent. 

Because we can restrict our transformations of the matrices in question
to the controllable space, it turns out that every matrix in question 
has a full Krylov sequence. (In particular, the characteristic polynomial
and annihilating polynomial are the same for such matrices.)

Let us consider a system described by the equations \ref{eqn:state111}
and \ref{eqn:state222}.
The classical state feedback of \cite{Wonham1967}  can be thought
of as a modification of a full Krylov sequence $(w^0, \cdots , w^n)$ 
to another  $(v^0, \cdots , v^n),$
through the rule
\begin{align}
\label{eqn:basic_eqnintro}
\ppmatrix{
v^0|& \cdots &|v^{n} 
}
\bbmatrix{ 1  & \sigma_{1}& \cdots & \cdots &\sigma_{n}\\
0&1 & \cdots & \cdots &\sigma_{n-1}\\
\vdots & \vdots &\vdots &\vdots& \vdots  \\
0&  0 & \cdots & 0& 1 
}
=\ppmatrix{
w^0|& \cdots &|w^{n} 
}.
\end{align}
If $s^n+ d_1s^{n-1}+\cdots +d_n,$
$s^n+  b_1s^{n-1}+\cdots +b_n,$
are respectively the annihilating polynomials of the matrices
with full Krylov sequences $(w^0, \cdots , w^n),$
 $(v^0, \cdots , v^n),$
then  we must have 
\begin{align}
\label{eqn:last_columnintro}
\bbmatrix{ 1  & \sigma_{1}& \cdots & \cdots &\sigma_{n}\\
0&1 & \cdots & \cdots &\sigma_{n-1}\\
\vdots & \vdots &\vdots &\vdots& \vdots  \\
0&  0 & \cdots & 0& 1 
}
\ppmatrix{
d_{n}\\   d_{n-1} \\\vdots \\1 
} = \ppmatrix{
b_{n}\\   b_{n-1} \\\vdots \\1 
}.
\end{align}
Further, 
one can show that the condition in equation \ref{eqn:last_columnintro}
is actually  equivalent to the one in  equation \ref{eqn:basic_eqnintro}.

The above corresponds to the matrix $A$ being transformed through
state feedback to $A+BF,$ where the state feedback matrix $F$ has rank one.
The output feedback transforms the matrix $A$ to the matrix $A+BKC,$
where $K$ is the output feedback matrix. 
When $K$ has rank one it corresponds to those state feedbacks 
where $(\sigma _1, \cdots , \sigma _n)\in row(C[w^0| \cdots |w^{n-1}])$
(Corollary \ref{cor:mainII}).
This essentially characterizes all matrices of the form $A+BKC,$
`reachable'
from $(A,B,C)$ through rank one output feedback. 

Coefficient matrix of the kind in equation \ref{eqn:last_columnintro}
is a triangular Toeplitz matrix and 
is entirely determined by the sequence $(1,\sigma _1, \cdots , \sigma _n).$
It is conveniently manipulated treating it as though it is the  polynomial
$s^n+ \sigma _1s^{n-1}+ \cdots + \sigma _n.$
Multiplication of such matrices corresponds to multiplying the 
corresponding polynomials modulo $s^{n+1}.$ One of the useful consequences
is that equation \ref{eqn:last_columnintro}
can be written as 
\begin{align}
\label{eqn:last_columnintro}
\bbmatrix{ 1  & d_{1}& \cdots & \cdots &d_{n}\\
0&1 & \cdots & \cdots &d_{n-1}\\
\vdots & \vdots &\vdots &\vdots& \vdots  \\
0&  0 & \cdots & 0& 1 
}
\ppmatrix{
\sigma_{n}\\   \sigma_{n-1} \\\vdots \\1 
} = \ppmatrix{
b_{n}\\   b_{n-1} \\\vdots \\1 
}.
\end{align}
A least square solution to this equation under the 
condition that $(\sigma _1, \cdots , \sigma _n)\in row(C[w^0| \cdots |w^{n-1}])$
can be used
to build a sequence $(1,\sigma _1, \cdots , \sigma _n)$
and transform the full Krylov sequence 
$(w^0, \cdots , w^n)$
to $(v^0, \cdots , v^n).$
The latter will have an annihilating polynomial closest in $l_2$
norm of coefficient vectors to $s^n+ b _1s^{n-1}+ \cdots + b_n.$
We call this the rank one update to the original annihilating
polynomial. These ideas can be used in practice to handle
the multi input - multi output case also.

The outline of the paper in terms of sections follows.

Section \ref{sec:Preliminaries}
is on preliminary notation and notions.  It also contains the statement of the useful `implicit inversion theorem'
 (Theorem \ref{thm:inverse}).

Section \ref{sec:outputfb} uses Theorem \ref{thm:inverse}
to give a simple characterization of matrices of the kind $A+BKC$
(Theorem \ref{thm:output}).

Section \ref{sec:rank_one} solves the output feed back 
problem for the case 
 where matrix
$B$ has a single column (Theorem \ref{thm:mainII}).

Section \ref{sec:mimo} presents a practical approach to  handling
multi-input multi-output case using the rank one update method repeatedly
(Algorithm I).

Section \ref{sec:conclusion} presents conclusions.

The Appendix contains a proof of a very general version 
of the implicit inversion theorem,
proofs of Lemma \ref{lem:1}, Lemma \ref{lem:2} 
and also the results of some numerical experiments.
These latter indicate that Algorithm I works in practice
and also that one can use the methods of this paper to alter
eigenvalues to desired locations.


\section{Preliminaries}
\label{sec:Preliminaries}
The preliminary results and the notation used are from \cite{HNarayanan1997} (open 2nd edition available at \cite{HNarayanan2009}).
They are needed primarily for the statement of Theorem \ref{thm:inverse}
and for the characterization of output feedback 
in Theorem \ref{thm:output}.

A \nw{vector} $\mnw{f}$ on $X$ over $\mathbb{F}$ is a function $f:X\rightarrow \mathbb{F}$ where $\mathbb{F}$ is a field. In this paper we work primarily with the real field.  When $X$, $Y$ are disjoint, a vector $f_{X\uplus  Y}$ on $X\uplus Y$ would be written as $f_{XY}$ and would often be written as $(f_X,f_Y)$ during operations dealing with such vectors. The {\bf sets} on which vectors are defined will  always be {\bf finite}. When a vector $x$ figures in an equation, we will use the 
convention that $x$ denotes a column vector and $x^T$ denotes a row vector such as
in $Ax=b,x^TA=b^T.$ Let $f_Y$ be a vector on $Y$ and let $X \subseteq Y$ . Then the \textbf{restriction} of $f_Y$ to $X$ is defined as follows:
\begin{align*}
f_Y/X \equiv g_X, \textrm{ where } g_X(e) = f_Y(e), e\in X.
\end{align*}
When $f$ is on $X$ over $\mathbb{F}$, $\lambda \in \mathbb{F}$ then  $\mnw{\lambda f}$ is on $X$ and is defined by $(\lambda f)(e) \equiv \lambda [f(e)]$, $e\in X$. When $f$ is on $X$ and $g$ on $Y$ and both are over $\mathbb{F}$, we define $\mnw{f+g}$ on $X\cup Y$ by 
\begin{align*}
 (f+g)(e) &\equiv \left\{ \begin{matrix}
                          f(e) + g(e),& e\in X \cap Y\\
			    f(e),& e\in X \setminus Y\\
			    g(e),& e\in Y \setminus X.
                         \end{matrix}
\right.
\end{align*}

A { collection} of vectors on $X$ over $\mathbb{F}$ 
is a \nw{vector space} iff $f,g\in \mathcal{K}$ implies $\lambda f + \sigma g \in \mathcal{K}$ for $\lambda, \sigma \in \mathbb{F}$. We will use the symbol $\mnw{\V_X}$ for vector space on $X.$ 

The symbol $\mnw{\mathcal{F}_X}$ refers to the \nw{collection of all vectors} on $X$ over $\mathbb{F}$ and $\mnw{0_X}$ to the \nw{zero vector space} on $X$ as well as the \nw{zero vector} on $X$. 
When $X_1,\cdots  ,X_r$ are disjoint we usually write $\mnw{\mathcal{V}_{X_1\cdots X_r}}$ in place of $\mnw{\mathcal{V}_{X_1\uplus \cdots \uplus X_r}}$.
The collection
$\{ (f_{X},\lambda f_Y) : (f_{X},f_Y)\in \mathcal{V}_{XY} \}$
is denoted by
$ \mathcal{V}_{X(\lambda Y)}.
$
When $\lambda = -1$ we write more simply $\mathcal{V}_{X(-Y)}.$


The \nw{sum} $\mnw{\mathcal{V}_X+\mathcal{V}_Y}$ of $\mathcal{V}_X$, $\mathcal{V}_Y$ is defined over $X\cup Y$ as follows:
$$ \mathcal{V}_X + \mathcal{V}_Y \equiv  \{  (f_X,0_{Y\setminus X}) + (0_{X\setminus Y},g_Y), \textrm{ where } f_X\in \mathcal{V}_X, f_Y\in \mathcal{V}_Y \},$$
When $X$, $Y$ are disjoint, $\mathcal{V}_X + \mathcal{V}_Y$ is usually written as $\mnw{\mathcal{V}_X \oplus \mathcal{V}_Y}$ and is called the \nw{direct sum}.

The \nw{intersection} $\mnw{\mathcal{V}_X \cap \mathcal{V}_Y}$ of $\mathcal{V}_X$, $\mathcal{V}_Y$ is defined over $X\cup Y$ as follows:
\begin{align*}
\begin{split}
\mathcal{V}_X \cap \mathcal{V}_Y \equiv \{ f_{X\cup Y} &: f_{X\cup Y} = (f_X,x_{Y\setminus X}),\qquad f_{X\cup Y} = (y_{X\setminus Y},f_Y), \\& \textrm{ where } f_X\in\mathcal{V}_X,\qquad f_Y\in\mathcal{V}_Y,   \qquad x_{Y\setminus X},\ y_{X\setminus Y} \\&\textrm{ are arbitrary vectors on }  Y\setminus X,\ X\setminus Y \textrm{ respectively} \}.
\end{split}
\end{align*}
Thus,
$$\mathcal{V}_X \cap \mathcal{V}_Y\equiv (\mathcal{V}_X \oplus \mathcal{F}_{Y\setminus X}) \cap (\mathcal{F}_{X\setminus Y} \oplus \mathcal{V}_Y).$$
The \nw{matched composition} $\mnw{\mathcal{V}_X \leftrightarrow \mathcal{V}_Y}$ is on $(X\setminus Y)\uplus (Y \setminus X)$ and is defined as follows:
\begin{align*}
 \mathcal{V}_X \leftrightarrow \mathcal{V}_Y &\equiv \{
                f : f=(g_{X\setminus Y} , h_{Y \setminus X}), \textrm{ where } g\in \mathcal{V}_X, h\in \mathcal{V}_Y \textrm{ \& } g_{X\cap Y} = h_{Y \cap X}
\}.
\end{align*}
Matched composition can be alternatively written using the above notation
as $$\mathcal{V}_X \leftrightarrow \mathcal{V}_Y \equiv 
(\mathcal{V}_X\cap\mathcal{V}_Y)\circ [(X\setminus Y)\uplus (Y \setminus X)]
= (\mathcal{V}_X+\mathcal{V}_{-Y})\times [(X\setminus Y)\uplus (Y \setminus X)].
$$
Matched composition is referred to as matched sum in \cite{HNarayanan1997}.
In the special case where $Y\subseteq X$, matched composition is called 
\nw{generalized minor} operation (generalized minor of $\mathcal{V}_X $
with respect to $\mathcal{V}_Y$). 

When $X$, $Y$ are disjoint, the matched  composition  corresponds to direct sum.
The operations $\mnw{\mathcal{V}_{XY} \leftrightarrow \mathcal{F}_Y}$, $\mnw{\mathcal{V}_{XY}\leftrightarrow \0_Y}$ are called, respectively, the \nw{restriction and contraction} of $\mathcal{V}_{XY}$ to $X$ and are also denoted by $\mnw{\mathcal{V}_{XY}\circ X}$, $\mnw{\mathcal{V}_{XY} \times X}$, respectively.
Here again $\mnw{\mathcal{V}_{XYZ}\circ XY}$, $\mnw{\mathcal{V}_{XYZ} \times XY}$, respectively
when $X,Y,Z$ are pairwise disjoint, would denote  $\mnw{\mathcal{V}_{XYZ}\circ (X\uplus Y)}$, $\mnw{\mathcal{V}_{XYZ} \times (X \uplus Y)}.$

It is easy to see that  $\mathcal{V}_{XYZ}\circ XY\circ X= \mathcal{V}_{XYZ}\circ X$ and $\mathcal{V}_{XYZ}\times XY\times X= \mathcal{V}_{XYZ}\times X.$


We have often made informal statements of the kind `a vector in 
$\Re^n$ has a certain property with probability $1.$' 

Here we assume that we work within a ball or cube in $\Re^n$ and the
  probability  of a measurable
 subset of that set is proportional to its lebesgue measure.


%

\subsubsection{Implicit Inversion Theorem}
The Implicit Inversion Theorem presents, in a sense, the most general solution to  the problem of computing, given $AB=C,$ the matrix $A$
in terms of $B$ and $C.$ This result along with the 
Implicit Duality Theorem (which is a generalization of `$AB=C \ \ \implies
B^TA^T= C^T$') may be regarded as the basic results of what may be called
`implicit linear algebra' \cite{HN2016}. They were originally derived in \cite{HNarayanan1986a},\cite{narayanan1987topological}. A comprehensive account is in 
\cite{HNarayanan1997}. Some additional applications are available in \cite{HNPS2013}
and \cite{HN2016}.

\begin{theorem}
\label{thm:inverse}
{\bf Implicit
Inversion Theorem}
 For the equation $\Vsp \lrar \Vpq =\Vsq $ 
\begin{enumerate}
\item given $\Vsp, \Vsq,  \exists \Vpq, $ satisfying the equation iff $\Vsp\circ S\supseteq \Vsq \circ S$ and $\Vsp\times S\subseteq \Vsq \times S;$
\item  given $\Vsp, \Vpq,\Vsq$ satisfying the equation, we have $\Vsp \lrar \Vsq =\Vpq $ iff
$\Vsp\circ P\supseteq \Vpq \circ P$ and $\Vsp\times P\subseteq \Vpq \times P.$
\item given $\Vsp, \Vsq,$ assuming that the equation  $\Vsp \lrar \Vpq =\Vsq $
is satisfied by some $\Vpq $ it is unique under the condition $\Vsp\circ P\supseteq \Vpq \circ P$ and $\Vsp\times P\subseteq \Vpq \times P.$
\end{enumerate}

\end{theorem}
It can be shown that the special case where $Q=\emptyset$ is equivalent to the above result. For convenience we state it below.
\begin{theorem}
\label{thm:Vgenminor}
Let $\Vs = \Vsp \lrar \Vp.$  Then, $\Vp = \Vsp \lrar \Vs$ iff
$\Vsp \circ P \supseteq \Vp$ and $\Vsp \times P \subseteq \Vp.$
\end{theorem}
\section{Characterization of output feedback}
\label{sec:outputfb}
In this section we give a characterization of  operators which can be obtained
from a given operator by output feedback. This characterization is used 
in the main result (Theorem \ref{thm:mainII}) later.
Let
\begin{align}
\label{eqn:state11}
 \mydot{w} &= Aw + Bu 
\end{align}
\begin{align}
\label{eqn:state22}
 y &= Cw .
\end{align}
be a set of equations where $A,B,C$ respectively are $ n\times n,n\times m$ and $p\times n$
real matrices.
Let $\Vwdwuy$ be the 
solution space of 
equations \ref{eqn:state11} and
\ref{eqn:state22} with its typical member being denoted by
$(w,\dot{w},u,y).$
We treat output feedback as a constraint on vectors $(u,y)$ to belong
to a  vector space $\Vuy$ and then, under the condition that rows of $C$
and columns of $B$ be linearly independent, show that the constraint 
can be taken to be $u=Ky.$ 
(This is in line with the point of view 
of the `behaviourists'
\cite{willems1991paradigms}
.) 
The natural way to proceed when we 
treat feedback in this general manner is to use the Implicit Inversion 
Theorem.
\begin{theorem}
\label{thm:output}
Let $A,B,C$ be matrices as in equations \ref{eqn:state11} and \ref{eqn:state22}.
Further let the columns of $B$ and the rows of $C$ be linearly 
independent.
Let $\hat{A}$ be a real $n\times n$ matrix and let $\Vwdw ^2$ be 
the solution space of 
 the equation
$\dot{w}=\hat{A}w.$ 
Then
\begin{enumerate}
\item $\hat{A}= A+BKC$ for some $m\times p$ matrix $K$ iff
\begin{enumerate}
\item for every $w\in \Re^n,$ $\hat{A}w=Aw+Bu$ for some $u\in \Re^m,$
and 
\item $\hat{A}w=Aw$ whenever $Cw=0.$
\end{enumerate}
Equivalently
\item $\Vwdw ^2=\Vwdwuy \leftrightarrow \Vuy,$ where $\Vuy$ is the solution space
of the equation $u=Ky$ for some $K$ iff
\begin{enumerate}
\item $\Vwdwuy\circ \wdw \supseteq \Vwdw ^2$ and
\item $\Vwdwuy\times \wdw \subseteq \Vwdw ^2 .$
\end{enumerate}
\end{enumerate}
\end{theorem}
The following lemmas are needed in the proof of Theorem \ref{thm:output}.
\begin{lemma}
\label{lem:A+BKC}
Let $\Vuy$ be the solution space of the equation $u=Ky.$
Then $\Vwdw^2= \Vwdwuy\leftrightarrow \Vuy$ iff
$\hat{A}=A+BKC.$
\begin{proof}
The matched composition $\Vwdwuy\leftrightarrow \Vuy$ is obtained by setting $u=Ky$
in the solution space of equations \ref{eqn:state11}
and \ref{eqn:state22}.
This results in the solution space of 
$\mydot{w}=Aw+BKCw,$
i.e., of the equation $\mydot{w}=(A+BKC)w.$
\end{proof}

\end{lemma}
\begin{lemma}
\label{lem:o}
\begin{enumerate}
In the statement of Theorem \ref{thm:output}
\item Conditions  $1(a),2(a)$ 
are equivalent.
\item Conditions $1(b),2(b)$ 
are equivalent.
\end{enumerate}
\end{lemma}
\begin{proof}
\begin{enumerate}
\item 
Observe that  $\Vwdwuy\circ W\dw U$ is the solution space of $\mydot{w} =  Aw + Bu .$ 

Let condition 1(a) be satisfied and let $(w,\dsw)$ be s.t 
$\mydot{w} = \hat{A}w= Aw + Bu $ for some $u.$ This means $(w,\dsw)$ belongs to  $\Vwdw^2$
only if it belongs to $(\Vwdwuy\circ W\dw U)\circ W\dw = \Vwdwuy\circ W\dw,$ i.e., only if $\Vwdwuy\circ W\dw \supseteq  \Vwdw^2 .$
Thus condition $1(a)$ implies condition $2(a).$

Next, suppose   $\Vwdwuy\circ W\dw \supseteq  \Vwdw^2 ,$
i.e., $(\Vwdwuy\circ W\dw U)\circ W\dw \supseteq  \Vwdw^2 .$
This means whenever $(w,\dsw)$ belongs to  $\Vwdw^2,$ 
i.e., satisfies $\dsw=\hat{A}w,$
there must exist some $u$ such that $\mydot{w} = Aw + Bu .$
Thus condition $2(a)$ implies condition $1(a).$
\item Let condition 1(b) be satisfied. The space $\Vwdwuy\times \wdw$ is the collection of all $(w,\dsw)$ obtained 
from equations \ref{eqn:state11}, \ref{eqn:state22}  by setting 
$u,y$ to zero, i.e., by setting $u=0$ in equation \ref{eqn:state11} and $y=0$ 
in equation \ref{eqn:state22}. 

Now  setting $u=0$ in equation \ref{eqn:state11} is the same as writing
$\dot{w}={A}w$ and setting $y=0$ 
in equation \ref{eqn:state22} is the same as setting $Cw=0.$ 

So if $\hat{A}w=Aw$ whenever $Cw=0,$ 
 we must have that every member of 
$\Vwdwuy\times \wdw$ also satisfies $\dsw=\hat{A}w,$
i.e., belongs to $\Vwdw^2.$ Thus condition $1(b)$ implies condition 
$2(b).$

Next suppose $\Vwdwuy\times \wdw \subseteq \Vwdw^2.$ So when $u,y$ are set to zero in equations \ref{eqn:state11}, \ref{eqn:state22}, we must have
$(w,\dsw)$ belonging to $\Vwdw^2,$ i.e., satisfying $\dsw=\hat{A}w.$  When $u$ is set to zero
we have $\dsw=Aw$ and when $y$ is set to zero we have $Cw=0.$
It follows that when $Cw=0,$ we must have $Aw=\hat{A}w.$
 Thus condition $2(b)$ implies condition
$1(b).$

\end{enumerate}
\end{proof}

Let $\V_{ABCD}$ be a vector space over a field $\mathbb{F}$ with typical vectors denoted $(a,b,c,d),$ corresponding to index sets $A,B,C,D.$
We say that the variable $a$ is {\bf free} in  $\V_{ABCD}$  iff
$\V_{ABCD}\circ A=\F_A,$ i.e., iff the variable $a$ can acquire every possible
vector value on $A.$

The routine proofs of the following lemmas  have been relegated to the Appendix.
\begin{lemma}
\label{lem:1}
Let $\Vwdwuy$ be the solution space of the equations \ref{eqn:state11} and
\ref{eqn:state22} and let $\Vwdw^2$ be the 
solution space of 
the equation
$\dot{w}=\hat{A}w.$
Let $\Vwdwuy\circ \wdw \supseteq \Vwdw^2$ and let the rows of $C$ 
be linearly independent.
Then 
\begin{enumerate}
\item The variable  $y$ is free in 
$(\Vwdwuy \cap  \Vwdw^2)\circ UY,$
\item $(\Vwdwuy\leftrightarrow \Vwdw^2)\circ Y=\F_Y.$ 
\end{enumerate}
\end{lemma}
The second part of the following lemma is the dual of the second part of Lemma \ref{lem:1}.
\begin{lemma}
\label{lem:2}
Let $\Vwdwuy$ be the solution space of the equations \ref{eqn:state11} and
\ref{eqn:state22} and let $\Vwdw^2$ be the
solution space of
the equation
$\dot{w}=\hat{A}w.$
Let $\Vwdwuy\times \wdw \subseteq \Vwdw^2$ and let the columns of $B$
be linearly independent.
Then
\begin{enumerate}
\item
In the solution space of the
Equations \ref{eqn:state11}, \ref{eqn:state22}
we must have $u=0$ whenever $y=0.$  

Equivalently
$\Vwdwuy\times W\dw U= \Vwdwuy\times W\dw\oplus \0_U.$


\item
$(\Vwdwuy\leftrightarrow \Vwdw^2)\times U=\0_U.$
\end{enumerate}
\end{lemma}

{\bf{Proof of Theorem \ref{thm:output}}}\\
(Necessity) Suppose $\hat{A}\equiv A+BKC,$ equivalently,
$\Vwdw^2=\Vwdwuy\leftrightarrow \Vuy,$ where $\Vuy$ is the solution space
of $u=Ky$ for some $K.$

We have $\hat{A}w=Aw+BKCw.$ Setting $u=KCw,$ we see that $\hat{A}w=Aw+Bu.$

Next if $Cw=0,$ $\hat{A}w=Aw+BKCw=Aw.$ Thus conditions 1(a) and 1(b)
are satisfied. By Lemma \ref{lem:o}, we must have that conditions 2(a) and 2(b)
are also satisfied.

(Sufficiency)  By Lemma \ref{lem:o},
for proving 
sufficiency, it is adequate to prove that when 2(a),2(b) are satisfied 
$\hat{A}$ can be written as $A+BKC$ for some matrix $K.$

By Theorem \ref{thm:inverse},
we have  that\\
 if $\Vwdwuy\circ \wdw \supseteq \Vwdw ^2$
and $\Vwdwuy\times \wdw \subseteq \Vwdw ^2$
then $\Vwdwuy \leftrightarrow \hat{\Vuy}= \Vwdw ^2,$\\
where $\hat{\Vuy}\equiv \Vwdwuy \leftrightarrow \Vwdw ^2.$

We will show that $\hat{\Vuy}$ is the solution space of $u=Ky,$ for some matrix $K.$ 

By Lemma \ref{lem:A+BKC}, it would follow  that
$\hat{A}= A+BKC.$

A general representation for $\hat{\Vuy},$ which can be obtained by routine
row operations, is as the solution space of equations

\begin{align}
\label{mat:vuy}
\bbmatrix{ R_{11}  & 0\\
R_{21} & R_{22} \\
0 & R_{32}
}
\ppmatrix{
        {u} \\ y
} = \ppmatrix{0\\0\\0},
\end{align}
where columns  of $(R_{11}^T,R_{21}^T)$ and $(R_{22}^T,R_{32}^T)$
are linearly independent.

By Lemmas \ref{lem:1}, \ref{lem:2}, when columns of $B$ and rows of $C$ are linearly independent and conditions 1(a), 1(b) (or 2(a),2(b)) are satisfied, we must have 
 
$(\Vwdwuy\leftrightarrow \Vwdw^2)\circ Y=\hat{\Vuy}\circ Y=\F_Y$
and $(\Vwdwuy\leftrightarrow \Vwdw^2)\times U=\hat{\Vuy}\times U=\0_U.$

But the former implies that $R_{32}$ does not exist and the latter, that
$R_{21}$ is nonsingular.

Thus we may write $u=((-R_{21})^{-1}R_{22})y.$

Taking $K\equiv ((-R_{21})^{-1}R_{22})$ yields the required result.
\qedsymbl
\begin{remark}
Suppose rows of $B$ and columns of $C$ are not linearly independent.
Let $U_1$ correspond to a maximal independent set of columns $B_1$ of $B$
and let $Y_1$  correspond to a maximal independent set of rows $C_1$ of $C.$
Let 
$\Vwdwuy\circ UY\supseteq \Vuy$ and   $\Vwdwuy\times UY\subseteq \Vuy.$ (This condition is satisfied by
$\hat{\Vuy}$ in the proof of Theorem \ref{thm:output}.)
Now if $\Vwdw= \Vwdwuy\leftrightarrow \Vuy,$
we can show 
that (see Problem $7.4$  in \cite{HNarayanan1997}, \cite{HNarayanan2009})$$\Vwdw= (\Vwdwuy\circ W\dw UY_1\times W\dw U_1Y_1)\leftrightarrow (\Vuy\circ  UY_1\times  U_1Y_1).$$
Now if we let $U_1$ correspond to a maximal independent set of columns $B_1$ of $B$
and let $Y_1$  correspond to a maximal independent set of rows $C_1$ of $C,$
this corresponds to working with the reduced system
$$ \mydot{w}\  = Aw + B_1u_1$$ 
$$ y_1 = C_1w .$$
The above result states, in particular, that if 
$\hat{A}=A+BKC,$ then we can find $K_1$ such that
$\hat{A}=A+B_1K_1C_1,$
So we lose no generality in assuming the linear independence of 
columns of $B$ and rows of $C.$
\end{remark}
\section{The rank one output feedback case}
\label{sec:rank_one}
In this section we solve the problem of characterizing annihilating polynomials of matrices of the form $A+BKC,$ 
where $K$ is a rank one matrix 
and $A,B,C$ are $n\times n,n\times m,p\times n$ matrices as in equations \ref{eqn:state11} and \ref{eqn:state22}.
Our methods are valid when the matrices are defined over arbitrary fields,
but our primary interest remains in real matrices.
\subsection{Preliminary assumptions}
\label{subsec:prelim}
As is well known, state feedback and therefore output feedback, cannot
take a state from within the controllable space to outside the space.
 We therefore lose no generality in assuming that the system we deal with
is fully controllable, i.e., the column space of
$[\ B|AB|\cdots |A^{n-1}B\ ]$ is $\Re ^n.$
Because of full controllability, 
 the annihilating polynomial of the matrix
$A$ in equation \ref{eqn:state11} can also be taken as its
characteristic polynomial.
Also when $B$ has only a single column, say $b,$ we may take
$w^0\equiv b$ and $\{ w^0,Aw^0,\cdots ,A^{n-1}w^0\}$ to be linearly independent.
When the matrix $B$ has more than one column, it is well known that
if we pick a vector $w^0$ in $col(B)$ at random, with probability $1,$
the above linear independence would hold.
(For a description of the notions of controllability, observability etc.
and for basic results of the kind mentioned above,
the reader may refer to \cite{Trentelman2001}, \cite{Wonham1978}.)

\subsection{The main result}
\begin{definition}
A sequence $\omega \equiv (w^0, \cdots , w^n)$ of $n\times 1$ vectors is said to be a full Krylov sequence for an $n\times n$ matrix $A,$
iff $w^j=A^jw^0, j= 1, \cdots n$ and $w^0, \cdots , w^{n-1}$
are linearly independent.
The vector $w^0$ is said to be the initial vector of ${\omega}.$
\end{definition}
If an $n\times n$ real matrix $A$ has its annihilating polynomial of degree $n$
then it has a full Krylov sequence. Further, 
it is clear that there is atmost one full Krylov sequence for a matrix
with a given initial vector. We note that  a full Krylov sequence uniquely defines the
matrix $A,$ since we must have $A(w^0| \cdots |w^{n-1})=(w^1| \cdots |w^{n})$
and $(w^0| \cdots |w^{n-1})$ has independent columns and therefore is  
nonsingular.
\begin{definition}
\label{def:sigma}
$\sigma^{(k+1)}$ denotes the Toeplitz upper triangular matrix 
\begin{align}
\bbmatrix{ \sigma_{0}  & \sigma_{1}& \cdots & \cdots &\sigma_{k}\\
0&\sigma_{0} &  \sigma_{1}  & \cdots &\sigma_{k-1}\\
\vdots & \vdots &\vdots &\vdots& \vdots  \\
0&  0 & \cdots & \sigma_{0} & \sigma_{1}\\
0&  0 & \cdots & 0& \sigma_{0}
}
\end{align}
We will take $\sigma_{0}\equiv 1.$

\end{definition}
\begin{theorem}
\label{thm:mainII}
Let
\begin{align}
\label{eqn:state1}
 \mydot{w} &= Aw + Bu \\
\label{eqn:state2}
 y &= Cw .
\end{align}
be a set of equations where $A,B,C$ respectively are $ n\times n,n\times m$ and $p\times n$
real matrices with the columns of $B$ and the rows of $C$ being linearly
independent.\\
Let $(w^0, \cdots ,w^n)$ be a full Krylov sequence for $A,$ with $w^0\in col(B)$
and let $(v^0| \cdots |v^n)\sigma^{(n+1)}= (w^0| \cdots |w^n).$ 
We then have 
\begin{enumerate}
\item $(v^0, \cdots ,v^n)$ is a full Krylov sequence for $\hat{A},$ where
$\hat{A}$ is such that for $w\in \Re^n,$ 
$\hat{A}w\equiv Aw+w^0u, $ for some $u\in \Re^n.$
\item Let $\hat{A}$ be as in the statement of the previous part. Then $\hat{A}w=Aw$ for $Cw=0$ iff\\
$(\sigma_1, \cdots , \sigma _n)\in \ row(C(w^0| \cdots |w^{n-1})).$
\end{enumerate}

\begin{proof}
We first note that by the definition of  $\sigma^{(k+1)}$  and since $(v^0| \cdots |v^n)\sigma^{(n+1)}= (w^0| \cdots |w^n),$
we must have  $(v^0| \cdots |v^{n-1})\sigma^{(n)}= (w^0| \cdots |w^{n-1}),$
with  $v^0=w^0,$ since $\sigma_0=1.$\\

1. Let $w= (w^0| \cdots |w^{n-1})\beta. $ Then, \\
$$\hat{A}w=\hat{A}(w^0| \cdots |w^{n-1})\beta=\hat{A}(v^0| \cdots |v^{n-1})\sigma^{(n)}\beta
 $$ $$=(v^1| \cdots |v^{n})\sigma^{(n)}\beta=(v^0| \cdots |v^{n})\bbmatrix{\sigma_1, \cdots, \sigma_n\\\sigma^{(n)}}\beta - v^0(\sigma_1, \cdots , \sigma_n)\beta.
$$
$=(w^1| \cdots |w^{n})\beta - v^0(\sigma_1, \cdots , \sigma_n)\beta
$$ $$= Aw-w^0(\sigma_1, \cdots , \sigma_n)\beta= Aw + w^0u,$ 
where $$u\equiv -(\sigma_1, \cdots , \sigma_n)\beta= -(\sigma_1, \cdots , \sigma_n)(w^0| \cdots |w^{n-1})^{-1}w.$$ 

2. We have $(w^0| \cdots |w^{n})=(v^0| \cdots |v^{n})\sigma^{(n+1)},$ i.e., \\
\begin{align}
w^0&\ \ \ \ \ \ \ \ \ \ \ \ \ \ \ \ \ \ \ \ \ \ \ \ =&v^0&\ \\
w^1&\ \ \ \ \ \ \ \ \ \ \ \ \ \ \ \ \ \ \ \ \ \ \ \ =& v^1&+\sigma_1v^0\\
\vdots&\ \ \ \ \ \ \ \ \ \ \ \ \ \ \ \ \ \ \ \ \ \ \ \ \vdots&\vdots&\vdots\\
w^n&\ \ \ \ \ \ \ \ \ \ \ \ \ \ \ \ \ \ \ \ \ \ \ \ =&v^n&+\sigma_1v^{n-1}+\cdots + \sigma_nv^{0}.
\end{align}
We note that since $\hat{A}v^k= v^{k+1},\  k=0, \cdots , n-1$
and
$w^i=v^i+\sigma_1v^{i-1}+\cdots + \sigma_iv^{0},$ 
we must have $\hat{A}w^i= v^{i+1}+\sigma_1v^i+\cdots + \sigma_iv^{1}.$
 We then have 
$${A}w^i=w^{i+1}=\hat{A}w^i+  \sigma_{i+1}v^{0}=\hat{A}w^i+  \sigma_{i+1}w^{0} .$$
Thus, 
$$A(w^0| \cdots |w^{n-1})=\hat{A}(w^0| \cdots |w^{n-1})+w^{0}(\sigma_1, \cdots , \sigma _n).$$
Let $w\equiv (w^0| \cdots |w^{n-1})\beta .$ Then
$$Aw= A(w^0| \cdots |w^{n-1})\beta = \hat{A}(w^0| \cdots |w^{n-1})\beta +w^{0}(\sigma_1, \cdots , \sigma _n)\beta.
$$
Therefore, $$Aw=\hat{A}w\ \  \mbox{iff}$$
$w^{0}(\sigma_1, \cdots , \sigma _n)\beta=
w^{0}(\sigma_1, \cdots , \sigma _n)(w^0| \cdots |w^{n-1})^{-1}w=0,
$ i.e., iff (since $w^0\ne 0$),
$$ (\sigma_1, \cdots , \sigma _n)(w^0| \cdots |w^{n-1})^{-1}w=0.$$
$$\mbox{Thus,}\  Aw=\hat{A}w \ \mbox{for}\  Cw=0\  \mbox{iff}$$
$$ (\sigma_1, \cdots , \sigma _n)(w^0| \cdots |w^{n-1})^{-1}w=0 \ \mbox{whenever}\ Cw=0,$$
$$\mbox{i.e., iff }\ (\sigma_1, \cdots , \sigma _n)(w^0| \cdots |w^{n-1})^{-1}\in row[C],$$
$$\mbox{i.e., iff }\ (\sigma_1, \cdots , \sigma _n)\in row[C(w^0| \cdots |w^{n-1})].$$

%

\end{proof}
\end{theorem}
\begin{corollary}
\label{cor:mainII}
Let $A,\hat{A},B,C$ be as in the statement of Theorem \ref{thm:mainII}.
Let $w^0=B\mu.$ Then
\begin{enumerate}
\item 
$\hat{A}=A+BF,$ where $F=-\mu(\sigma_1, \cdots , \sigma _n)(w^0| \cdots |w^{n-1})^{-1};$
\item 
$\hat{A}=A+BF=A+BKC,$ 
if
$(\sigma_1, \cdots , \sigma _n)=\rho^TC(w^0| \cdots |w^{n-1})$
and
$K= -\mu \rho^T.$
\end{enumerate}
\end{corollary}
\begin{proof}
1. In the proof of the first part of Theorem \ref{thm:mainII} we saw that when
$w\equiv (w^0| \cdots |w^{n-1})\beta ,$ we have
$$\hat{A}w=Aw- w^{0}(\sigma_1, \cdots , \sigma _n)(w^0| \cdots |w^{n-1})^{-1}w.$$ Since $w^{0}=B\mu,$
it follows that  $$\hat{A}w= Aw -B\mu(\sigma_1, \cdots , \sigma _n)(w^0| \cdots |w^{n-1})^{-1}w.$$
Setting  $F\equiv -\mu (\sigma_1, \cdots , \sigma _n)(w^0| \cdots |w^{n-1})^{-1}w,$
we get the required result.

2. In the proof of the previous part we saw that
$$\hat{A}w= Aw -B\mu(\sigma_1, \cdots , \sigma _n)(w^0| \cdots |w^{n-1})^{-1}w.$$ Let $(\sigma_1, \cdots , \sigma _n)=\rho^TC((w^0| \cdots |w^{n-1})$. If $K\equiv -\mu \rho^T$
we have \\$- \mu(\sigma_1, \cdots , \sigma _n)(w^0| \cdots |w^{n-1})^{-1}= -\mu \rho^TC=KC$
and $\hat{A}w=(A+BKC)w.$

\end{proof}
\subsection{Algebra of Toeplitz triangular matrices}
In this section we show that Toeplitz triangular matrices of the kind $\sigma^{(k+1)}$ (as in 
Definition \ref {def:sigma}) for a fixed $k$
have a very simple algebra. This would be useful to verify whether
for a system defined through equations \ref{eqn:state11}, \ref{eqn:state22}, a matrix $\hat{A},$ with a  given annihilating polynomial can be obtained
by rank one output feedback from the  matrix $A.$ 

Let $\boldsymbol{\hat{\Sigma}}(k+1)$ denote the collection of length $(k+1)$ sequences $\boldsymbol{\sigma}\equiv (\sigma_0, \cdots , \sigma _k).$
We define addition and mutiplication for this collection through
$$\boldsymbol{\sigma}+\boldsymbol{\rho}\equiv (\sigma_0, \cdots , \sigma _k)+(\rho_0, \cdots , \rho _k)=(\sigma_0 +\rho_0, \cdots , \sigma_k +\rho _k) \ 
\mbox{and}$$ $$\boldsymbol{\sigma}*\boldsymbol{\rho}\equiv (\sigma_0 \rho_0,\cdots,
\Sigma _{j=0}^{j=i}\sigma_j\rho_{(i-j)}, \cdots , \Sigma _{j=0}^{j=k}\sigma_j\rho_{(k-j)}).$$
If $\boldsymbol{\sigma}(s)\equiv \sigma_0+ \sigma_1s+ \cdots + \sigma _ks^k$
and $\boldsymbol{\rho}(s)\equiv \rho_0+ \rho_1s+ \cdots + \rho _ks^k$
then it is clear that $$(\boldsymbol{\sigma}*\boldsymbol{\rho})(s)=
\boldsymbol{\sigma}(s)\times \boldsymbol{\rho}(s)\  modulo\  s^{k+1}.$$
It follows that $\boldsymbol{\sigma}*\boldsymbol{\rho}=\boldsymbol{\rho}*\boldsymbol{\sigma}$ and that  $`*$' is distributive with respect to $`+$' so that
$\boldsymbol{\hat{\Sigma}}(k+1)$ is a commutative ring under addition $`+$' 
and multiplication $`*$'.

Now let us associate with $\boldsymbol{\sigma}\equiv (\sigma_0, \cdots , \sigma _k)$ the matrix 
$\sigma^{(k+1)}$ which, as in Definition \ref{def:sigma}, denotes the upper triangular matrix
\begin{align}
\bbmatrix{ \sigma_{0}  & \sigma_{1}& \cdots & \cdots &\sigma_{k}\\
0&\sigma_{0} & \cdots & \cdots &\sigma_{k-1}\\
\vdots & \vdots &\vdots &\vdots& \vdots  \\
0&  0 & \cdots & 0& \sigma_{0}
}
\end{align}
If $\boldsymbol{\sigma}*\boldsymbol{\rho}=\boldsymbol{\mu},$
it is easily verified that 
$\sigma^{(k+1)}\times \rho^{(k+1)}=\rho^{(k+1)}\times \sigma^{(k+1)}=\mu^{(k+1)}.$

In the computations of interest to us the above upper triangular matrices
have diagonal elements equal to $1.$ Addition does not play an important part.
So we introduce the notation $\boldsymbol{\Sigma}(k+1)$ for  the collection of length $(k+1)$ sequences $\boldsymbol{\sigma}\equiv (\sigma_0, \cdots , \sigma _k),$
with $\sigma _0=1.$ 
The following theorem summarizes simple, but useful facts about $\boldsymbol{\Sigma}(k+1).$
The routine proof is omitted.
\begin{theorem}
\label{thm:sigma_properties}
\begin{enumerate}
\item The $(k+1)$ length sequence
$(1,0,\cdots, 0)$ acts as the identity for $\boldsymbol{\Sigma}(k+1)$ under the multiplication operation $`*$' . 
\item Every element of  $\boldsymbol{\Sigma}(k+1)$ has an inverse. 
\item $\boldsymbol{\Sigma}(k+1)$ is a commutative group under $`*$'.
\item When $\boldsymbol{\mu}\equiv (1,0,\cdots, 0),$
$\mu^{(k+1)}$ is the identity matrix of size $(k+1)\times (k+1)$ 
\item If $\boldsymbol{\sigma}*\boldsymbol{\rho}=(1,0,\cdots, 0),$
i.e, if $\boldsymbol{\sigma}$ and $\boldsymbol{\rho},$
are inverses of each other
under $`*$', the matrices $\sigma^{(k+1)}$ and $\rho^{(k+1)}$
are inverses of each other.
\item Let $\boldsymbol{\rho},\boldsymbol{\sigma} \in \boldsymbol{\Sigma}(k+1).$ If $\boldsymbol{\rho}*\boldsymbol{\sigma}= \boldsymbol{\mu},$
then, using the fact that
$\rho^{(k+1)}\times \sigma^{(k+1)}=\mu^{(k+1)},$ 
\begin{align}
\label{eqn:first_row}
\ppmatrix{ 
 \rho_{0}  & \rho_{1}& \cdots & \cdots &\rho_{k}}
\bbmatrix{ \sigma_{0}  & \sigma_{1}& \cdots & \cdots &\sigma_{k}\\
0&\sigma_{0} & \cdots & \cdots &\sigma_{k-1}\\
\vdots & \vdots &\vdots &\vdots& \vdots  \\
0&  0 & \cdots & 0& \sigma_{0}
}
=\ppmatrix{ 
 \mu_{0}  & \mu_{1}& \cdots & \cdots &\mu_{k}}.
\end{align}
\begin{align}
\label{eqn:last_column}
\bbmatrix{ \rho_{0}  & \rho_{1}& \cdots & \cdots &\rho_{k}\\
0&\rho_{0} & \cdots & \cdots &\rho_{k-1}\\
\vdots & \vdots &\vdots &\vdots& \vdots  \\
0&  0 & \cdots & 0& \rho_{0} 
}
\ppmatrix{
\sigma_{k}\\   \sigma_{1} \\\vdots \\ \sigma_{0}
} = \ppmatrix{
\mu_{k}\\   \mu_{1} \\\vdots \\ \mu_{0}
}.
\end{align}
\item If  equation \ref{eqn:first_row} or equation \ref{eqn:last_column}
is true, then  $\boldsymbol{\rho}*\boldsymbol{\sigma}= \boldsymbol{\sigma}*\boldsymbol{\rho}= \boldsymbol{\mu}.$
\end{enumerate}
\end{theorem}
We will denote the inverse of $\boldsymbol{\sigma}$ under $`*$'
by $\boldsymbol{\sigma}^{-1}$.

\subsection{Matrices and full Krylov sequences}
Our discussion of annihilating polynomials of matrices obtainable 
by output feedback 
from the matrix $A,$ 
will use their full Krylov sequences. In this section we summarize their
useful properties and derive a result on state feedback which
will help us solve the linear output feedback annihilation polynomial
problem when there is only a single input or a single output.

\begin{lemma}
\label{lem:unique}
Let $(w^0, \cdots ,w^{n})$ be a full Krylov sequence of the $n\times n $ matrix $A$.
Then the equation 
\begin{align}
\label{eqn:basic_eqn_unique}
\ppmatrix{
w^0|& \cdots &|w^{n} 
}\ppmatrix{
d_{n}\\   d_{n-1} \\\vdots \\ d_{0}
} 
=0
\end{align}
has a unique solution under the condition that $d_0=1.$
\end{lemma}
\begin{proof}
By definition of full Krylov sequence, the matrix $(w^0| \cdots |w^{n-1})$
has linearly independent columns while the matrix $(w^0| \cdots |w^{n})$
has dependent columns. Therefore 
the only solution of equation \ref{eqn:basic_eqn_unique},
when $d_0=0,$ is the zero solution 
and the solution space of the equation 
has dimension $1,$ being the collection of scalar multiples of a nonzero vector.
The result follows.  
\end{proof}
The next result relates the annihilating polynomials of 
matrices $A, A+BF,$ when $F$ has rank one (see part 1 of Corollary \ref{cor:mainII}).
\begin{theorem}
\label{thm:ann_relate}
Let $(w^0, \cdots ,w^{n}),(v^0, \cdots ,v^{n})$ be full Krylov sequences 
of $n\times n$ matrices $A,\hat{A}$ respectively.
Further let $(v^0| \cdots |v^n)\sigma^{(n+1)}=(w^0| \cdots |w^{n}).$
Let $d_0s^n+d_1s^{n-1}+\cdots +d_n, 
b_0s^n+b_1s^{n-1}+\cdots +b_n,$
be the annihilating polynomials of $A,\hat{A}$ respectively.
Let $\sigma_0=d_0=b_0=1.$\\
Then, 
\begin{enumerate}
\item
\begin{align}
\label{eqn:last_columndash}
\bbmatrix{ \sigma_{0}  & \sigma_{1}& \cdots & \cdots &\sigma_{n}\\
0&\sigma_{0} & \cdots & \cdots &\sigma_{n-1}\\
\vdots & \vdots &\vdots &\vdots& \vdots  \\
0&  0 & \cdots & 0& \sigma_{0} 
}
\ppmatrix{
d_{n}\\   d_{n-1} \\\vdots \\ d_{0}
} = \ppmatrix{
b_{n}\\   b_{n-1} \\\vdots \\ b_{0}
}.
\end{align}
Hence,
$\boldsymbol{\sigma}*\boldsymbol{d}=\boldsymbol{d}*\boldsymbol{\sigma}=\boldsymbol{b},
\mbox{where}\  \boldsymbol{\sigma},\boldsymbol{d},\boldsymbol{b},\ 
\mbox{denote the sequences} \ (\sigma_0, \cdots , \sigma_n),(d_0, \cdots , d_n),(b_0, \cdots , b_n),
\mbox{respectively.}$\\
\item
Conversely, if equation \ref{eqn:last_columndash}
holds,  and 
$(v^0| \cdots |v^n)\hat{\sigma}^{(n+1)}=(w^0| \cdots |w^{n}),$
then $\hat{\sigma}^{(n+1)}= {\sigma}^{(n+1)}.$
\end{enumerate}

\end{theorem}
\begin{proof}
1. We have
\begin{align}
\label{eqn:basic_eqn}
0=
\ppmatrix{
v^0|& \cdots &|v^{n} 
}\ppmatrix{
b_{n}\\   b_{1} \\\vdots \\ b_{0}
}= 
\ppmatrix{
w^0|& \cdots &|w^{n} 
}\ppmatrix{
d_{n}\\   d_{1} \\\vdots \\ d_{0}
} 
= [\ppmatrix{
v^0|& \cdots &|v^{n} 
}\ppmatrix{\sigma^{(n+1)}}]\ppmatrix{
d_{n}\\   d_{1} \\\vdots \\ d_{0}
}. 
\end{align}
But by Lemma \ref{lem:unique},
the equation 
\begin{align}
\ppmatrix{
v^0|& \cdots &|v^{n} 
}\ppmatrix{
b_{n}\\   b_{1} \\\vdots \\ b_{0}
}=0,
\end{align}
has a unique solution under the condition that $b_0=1.$
Now $\sigma_0d_0$ is the last entry of the column $\sigma^{(n+1)}
\ppmatrix{
d_{n}\\   d_{1} \\\vdots \\ d_{0}
}.
$
Since $\sigma_0,d_0 =1,$ we must have 
$\sigma_0d_0=1.$  

The result follows.

By Theorem \ref{thm:sigma_properties}, part 7, it follows that $\boldsymbol{\sigma}*\boldsymbol{d}=\boldsymbol{b}.$
\\

2. If equation \ref{eqn:last_columndash} holds, then by Theorem \ref{thm:sigma_properties}, part 7, we must have 
$\boldsymbol{\sigma}*\boldsymbol{d}=\boldsymbol{b},$ and if this equation
holds then by Theorem \ref{thm:sigma_properties}, $\boldsymbol{d}*\boldsymbol{b}^{-1}= \boldsymbol{\sigma}.$

But if 
$(v^0| \cdots |v^n)\hat{\sigma}^{(n+1)}=(w^0| \cdots |w^{n}),$
by the previous part, we must have $\boldsymbol{\hat{\sigma}}*\boldsymbol{d}=\boldsymbol{b},$ 
and therefore $\boldsymbol{d}*\boldsymbol{b}^{-1}= \boldsymbol{\hat{\sigma}}.$

\end{proof}

\subsection{Solution of the rank one output feedback annihilating polynomial problem}
\label{subsec:rank_one}
For the system given in equations \ref{eqn:state11} and \ref{eqn:state22},
let us suppose the matrix $B$ has a single column.
(If $C$ has a single row we could work with the dual system.)
Let $A$ have the annihilating polynomial $d_0s^n+ \cdots +d_n.$
Suppose the desired annihilating polynomial is $b_0s^n+ \cdots +b_n.$
We need to find, if it exists, a matrix $K,$ such that
$A+BKC$ has the desired annihilating polynomial, and if it does not exist,
prove that it does not.

Firstly, by Theorem \ref{thm:ann_relate}, we need to find a sequence
$\boldsymbol{\sigma}\equiv (\sigma_0, \cdots , \sigma _n)$
such that $\boldsymbol{\sigma}*\boldsymbol{d}=\boldsymbol{b}
,$ where $\boldsymbol{d}\equiv (d_0, \cdots , d _n),$
$\boldsymbol{b}\equiv (b_0, \cdots , b _n).$
This ensures rank one state feedback by Corollary \ref{cor:mainII}, part 1.
Next, by part 4 of the same theorem, we must have 
$(\sigma_1, \cdots , \sigma _n)\in \ row(C(w^0| \cdots |w^n)).$
These two conditions are together necessary and sufficient.

We have
$\boldsymbol{\sigma}*\boldsymbol{d}=\boldsymbol{b},$ i.e., 
\begin{align}
\label{eqn:extend_mainfb1}
\ppmatrix{ 
 \sigma_{0}  & \sigma_{1} & \cdots &\sigma_{n}}
\bbmatrix{ d_{0}  & d_{1} & \cdots &d_{n}\\
0&d_{0} & \cdots & d_{n-1}\\
\vdots & \vdots &\vdots & \vdots  \\
0&   \cdots & 0& d_{0}
}=
\ppmatrix{ 
 b_{0}  & b_{1} & \cdots &b_{n}}
.
\end{align}
We note that $\sigma_0,d_0,b_0$ all equal $1.$
Therefore, 
\begin{align}
\label{eqn:extend_mainfb1ia}
\ppmatrix{ 
{0}  & \sigma_{1} & \cdots &\sigma_{n}}
\bbmatrix{ d_{0}  & d_{1}& \cdots & \cdots &d_{n}\\
0&d_{0} & \cdots & \cdots &d_{n-1}\\
\vdots & \vdots &\vdots &\vdots& \vdots  \\
0&  0 & \cdots & 0& d_{0}
}=
\ppmatrix{ 
 b_{0}  & b_{1}&  \cdots &b_{n}}-\ppmatrix{ 
 d_{0}  & d_{1}&  \cdots &d_{n}}
,
\end{align}
i.e., 
\begin{align}
\label{eqn:extend_mainfb1ib}
\ppmatrix{ 
 \sigma_{1}&  \cdots &\sigma_{n}}
\bbmatrix{ d_{0}  & d_{1} & \cdots &d_{n-1}\\
0&d_{0} & \cdots & d_{n-2}\\
\vdots & \vdots &\vdots & \vdots  \\
0&   \cdots & 0& d_{0}
}=
\ppmatrix{ 
  b_{1}&  \cdots &b_{n}}-\ppmatrix{ 
  d_{1}&  \cdots &d_{n}},
\end{align}

%

Now by Corollary \ref{cor:mainII}, part 2, 
we must have 
\begin{align}
\ppmatrix{ 
  \sigma_{1}&  \cdots &\sigma_{n}}
=
\ppmatrix{ 
 \rho_{1}  & \cdots  &\rho_{p}} C
\ppmatrix{
w^0|& \cdots &|w^{n-1}}, 
\end{align}
for some $ (\rho_{1},   \cdots  ,\rho_{p}).$
Thus we need to solve the equation 
\begin{align}
\label{eqn:mainfb}
\ppmatrix{ 
\rho_{1}  & \cdots  &\rho_{p}}C\ppmatrix{
w^0|& \cdots &|w^{n-1}} 
\bbmatrix{ d_{0}  & d_{1} & \cdots &d_{n-1}\\
0&d_{0} & \cdots & d_{n-2}\\
\vdots & \vdots &\vdots & \vdots  \\
0&   \cdots & 0& d_{0}
}=
\ppmatrix{ 
  b_{1}&  \cdots &b_{n}}-\ppmatrix{ 
  d_{1}&  \cdots &d_{n}}.
\end{align}

As a preliminary to solving this equation, we can project 
right side vector of the equation \ref{eqn:mainfb}  onto the 
row space of 
\begin{align}
Q\equiv
C(w^0| \ \cdots \ |w^{n-1})
\bbmatrix{ d_{0}  & d_{1} & \cdots &d_{n-1}\\
0&d_{0} & \cdots & d_{n-2}\\
\vdots & \vdots &\vdots & \vdots  \\
0&   \cdots & 0& d_{0}
}.
\end{align}
If the projection equals the right side vector of the equation \ref{eqn:mainfb},
we compute 
$$(  \sigma_{1},  \cdots ,\sigma_{n})
=
 (\rho_{1}  , \cdots  ,\rho_{p}) C
(w^0| \cdots |w^{n-1}) 
$$
and thence 
$$(v^0| \cdots |v^{n}) = (w^0| \cdots |w^{n})(\sigma^{(n+1)})^{-1}.$$ 
By  Corollary \ref{cor:mainII}, part 2,
the matrix  with the full Krylov sequence $(v^0| \cdots |v^{n}) $ has the form $A+BKC,$
where $K=-(\rho_{1},   \cdots  ,\rho_{p}).$
This matrix has the annihilating polynomial $b_0s^n+ \cdots +b_n.$

Since the matrix $Q$
has rank $p\leq n,$
the equation \ref{eqn:mainfb} need not have a solution.
In this case the projection will not agree with the right side vector
and we can conclude that there is no matrix 
of the form $A+BKC$ whose annihilating polynomial is
$b_0s^n+\cdots +b_n.$

The discussion thus far is valid when $A,B,C$ are matrices
over arbitrary fields.

If we are dealing with matrices over the real field, 
the above procedure through projection
is the usual solution in the least square sense.
\begin{definition}
\label{def:rank_one}
Let $(w^0| \cdots |w^{n})$
be a full Krylov sequence of $A$ with $w^0\equiv B\mu.$
Let $(\rho_{1}'  , \cdots  ,\rho_{p}')$ be the least square solution 
to the equation \ref{eqn:mainfb} and let
$(\sigma_0',  \sigma_{1}',  \cdots ,\sigma_{n}')
$ with $\sigma_0'=1$ be such that
$$(  \sigma_{1}',  \cdots ,\sigma_{n}')
\equiv
 (\rho_{1}'  , \cdots  ,\rho_{p}') C
(w^0| \cdots |w^{n-1}) 
.$$
We will refer to $\boldsymbol{\sigma}'*\boldsymbol{d}$
as the rank one update to $\boldsymbol{d}$
for the desired $\boldsymbol{b}$
 and to $A+BKC$ with $K\equiv -\mu (\rho_{1}'  , \cdots  ,\rho_{p}')$
as the rank one update to $A$ for the desired annihilating polynomial
$b_0s^n+\cdots +b_n .$
\end{definition}

We now discuss the least square solution to equation \ref{eqn:mainfb}.
We state simple facts about this solution in the following lemma.
We have denoted the $l_2$ norm of a vector $x$ by $\|x\|.$

\begin{lemma}
\label{lem:lse}
Let 
\begin{align}
\label{eqn:lse}
{x^T}A={b^T},
\end{align}
where $A$ is an $m\times n$ matrix, be a linear equation.
Then, 
\begin{enumerate}
\item $b^T$ can be uniquely decomposed into $b'^T$ and $b"^T$
where $b'^T \in row(A)$ and $Ab"=0;$
\item $<b',b">=0;$ so $\|b'\|\leq \|b\|;$
\item the least square solution  $\hat{x}^T$ to equation \ref{eqn:lse}
 satisfies $\hat{x}^TA=b'^T ;$
if $(b^3)^T\equiv x^TA,$ for some $x,$ then
 $\|b-b'\|\leq \|b-b^3\|;$
\end{enumerate}
\end{lemma}
We now  have the following useful result.
\begin{theorem}
\label{thm:lse_equiv}
Let, as in Definition \ref{def:rank_one},
 $\boldsymbol{\sigma}'*\boldsymbol{d}$
be the rank one update to $\boldsymbol{d}$
for  $\boldsymbol{b}.$
Then 
\begin{enumerate}
\item
$\|\boldsymbol{b}-(\boldsymbol{\sigma}'*\boldsymbol{d})\|\leq \|\boldsymbol{b}- \boldsymbol{d}\|,$
with equality precisely when the least square solution to
equation \ref{eqn:mainfb} is the zero solution.
\item If $\boldsymbol{\sigma}$ is such that $(\sigma_1, \cdots , \sigma _n)
\in row(C(w^0| \cdots |w^{n}),$ then $\|\boldsymbol{b}-(\boldsymbol{\sigma}'*\boldsymbol{d})\|\leq \|\boldsymbol{b}-(\boldsymbol{\sigma}*\boldsymbol{d})\|.$
\end{enumerate}
\end{theorem}
\begin{proof}
1. For any $\boldsymbol{\sigma}",$ $(\boldsymbol{b}-(\boldsymbol{\sigma}"*d))$
can be written as $\boldsymbol{b}-\boldsymbol{d}-(0,\sigma_1", \cdots , \sigma_n")*d.$ In the present case, when $\boldsymbol{\sigma}'$ corresponds
to the least square solution of equation \ref{eqn:mainfb},
we have an orthogonal decomposition
of $\boldsymbol{b}-\boldsymbol{d}$ into $(0,\sigma_1', \cdots , \sigma_n')*\boldsymbol{d}$
and
$\boldsymbol{b}-\boldsymbol{d}-(0,\sigma_1', \cdots , \sigma_n')*\boldsymbol{d}.$
Since $\boldsymbol{b}-\boldsymbol{d}-(0,\sigma_1', \cdots , \sigma_n')*\boldsymbol{d}= \boldsymbol{b}-(\boldsymbol{\sigma}'*\boldsymbol{d}),$
it follows that $\|\boldsymbol{b}-(\boldsymbol{\sigma}'*\boldsymbol{d})\|\leq \|\boldsymbol{b}-\boldsymbol{d}\|.$
For the two norms to be equal, we need $(0,\sigma_1', \cdots , \sigma_n')*\boldsymbol{d}$
to be zero, i.e., the least square solution to equation \ref{eqn:mainfb}
to be zero.

2. This follows from part 3 of Lemma \ref{lem:lse}.
\end{proof}
\section{The multi input-multi output (MIMO) case.}
\label{sec:mimo}
The MIMO case is much harder than the single input or single output case.
In this section we suggest a practical approach to this problem
by breaking it into repeated rank one problems.

In the MIMO case, the matrix $A+BKC$ obtained by output feedback
can have the matrix $K$ of rank as high as $m,$ the number of columns
of $B.$ 
We make some preliminary assumptions taking the underlying
field to be $\Re^n.$

We remind the reader of the  assumption made in Subsection \ref{subsec:prelim},
that the controllable space for a system defined through equations \ref{eqn:state11}, \ref{eqn:state22}, is $\Re^n.$
If we pick a vector $w^0$ at random from $col(B),$
it can be shown, by standard arguments, that  with probability $1,$ the matrix $A$ would have, with $w^0$ as the initial vector, a full Krylov sequence $\omega \equiv (w^0, \cdots , w^n).$ 
If we pick 
$m$ vectors at random from $col(B),$ with probability
$1,$
they would be independent 
and for each such $w^0(i),$ we would get a full Krylov sequence with initial vector $w^0(i).$
Therefore we assume, without loss of generality, that 
the columns of $B$ are linearly independent and 
that these columns are initial vectors of  full Krylov sequences of $A.$

We then have the following elementary but useful result.
\begin{theorem}
Let the columns of $B, \ $ $w^0(1), \cdots w^0(m),$ 
be initial vectors of full Krylov sequences $\omega (1), \cdots , \omega (m),$
respectively of $A.$
\begin{enumerate}
\item Let $\omega $ be any full Krylov sequence of $A$ with initial vector in $col(B).$
Then $\omega \equiv (w^0, \cdots, w^n)$ is a linear combination of $\omega (i), i=1, \cdots m,$
i.e., there exist $\alpha_i , i=1, \cdots m$ s.t. 
$w^j=\Sigma _{i=1}^m \alpha_i w^j(i), j=0, \cdots , n.$
\item Let $\omega \equiv (w^0, \cdots, w^n)$  
be a linear combination of $\omega (i), i=1, \cdots m$
and let $w^0, \cdots, w^{n-1}$ be linearly independent.
Then $\omega$ is a full Krylov sequence of $A$ with initial vector
$\Sigma _{i=1}^m \alpha_i w^0(i)\in col(B).$
\end{enumerate}
\end{theorem}
\begin{proof}
1. We note that a full Krylov sequence $\omega \equiv (w^0, \cdots, w^n)$ 
of $A$ is defined by $w^j=A^{(j-1)}w^0, j=1, \cdots n.$
Therefore the initial vector of a full Krylov sequence of $A$ fully determines it.
The result now follows from the fact that  $w^0(1), \cdots w^0(m),$
form a basis for $col(B).$

2. Every linear  combination $\omega \equiv (w^0, \cdots, w^n) $ of $\omega (i), i=1, \cdots m,$
clearly satisfies $w^j=A^{(j-1)}w^0, j=1, \cdots n.$
It therefore is a full Krylov sequence of $A$ provided the first $n$ vectors
of the sequence are linearly independent.
\end{proof}

\begin{remark}
It is routine to note that every random linear combination of
 $\omega (i), i=1, \cdots m,$
will have its first $n$ vectors linearly independent, 
 and, therefore, be a full Krylov sequence of $A,$
with probability $1.$

\end{remark}

A natural question for the case where the number of columns of $B$ is greater
than one, is whether a desired annihilating polynomial can be achieved
by rank one update
using some vector of $col(B)$ as the initial vector of  a full Krylov sequence of
$A.$ We now reduce  this problem to the solution of equation \ref{eqn:mainfb2}
given below.

We begin with a matrix $B$ whose columns $w^0(1), \cdots ,w^0(m),$ are linearly independent and
form initial vectors of full Krylov sequences  $\omega (1), \cdots , \omega (m),$ 
of the matrix $A.$ 

Let $\omega \equiv (w^0, \cdots, w^n)$ be a linear combination of $\omega (i), i=1, \cdots m,$
Thus we have
$$(w^0, \cdots, w^n)=\Sigma _{i=1}^m\alpha_i(w^0(i), \cdots, w^n(i)).$$
Equation \ref{eqn:mainfb} now reduces to 
\begin{align}
\label{eqn:mainfb2}
\ppmatrix{ 
\rho_{1}  & \cdots  &\rho_{p}}C
\ppmatrix{
\Sigma _i\alpha_iw^0(i)|& \cdots &|\Sigma _i\alpha_iw^{n-1}(i)} 
\bbmatrix{ d_{0}  & d_{1} & \cdots &d_{n-1}\\
0&d_{0} & \cdots & d_{n-2}\\
\vdots & \vdots &\vdots & \vdots  \\
0&   \cdots & 0& d_{0}
}
\end{align}
\begin{align}
=
\ppmatrix{ 
  b_{1}& \cdots  &b_{n}}-\ppmatrix{ 
  d_{1}& \cdots  &d_{n}}.
\end{align}
If a solution exists to equation \ref{eqn:mainfb2}, then 
the value of $\alpha_i, i=1, \cdots m$
will yield the initial vector 
$w^0=\Sigma _i\alpha_iw^0(i)$
and thence the full Krylov sequence (provided the first $n$ columns are linearly independent) $$(w^0, \cdots, w^n)\equiv (\Sigma _i\alpha_iw^0(i)| \cdots |\Sigma _i\alpha_iw^{n}(i)).$$
As  in Subsection \ref{subsec:rank_one}, $\boldsymbol{\sigma}\equiv (1,\sigma_1, \cdots, \sigma_n),$ where $(\sigma_1, \cdots, \sigma_n)\equiv 
(\rho_{1}  , \cdots  ,\rho_{p})C (w^0, \cdots, w^n)$ 
and $K$ would be given by 
\begin{align} 
 \ppmatrix{- \alpha_1\\\vdots\\-\alpha_m}
\ppmatrix{ 
 \rho_{1}  & \cdots  &\rho_{p}}.
\end{align} 
\begin{remark}
Equation \ref{eqn:mainfb2} is a nonlinear equation 
where every term occurs as a product $\rho_i\alpha_j.$
A simple approach to solving this problem would be to start with 
an initial guess $\boldsymbol{\alpha }$ vector, compute the
$\boldsymbol{\rho }$ vector by the projection method,
then in the next iteration compute $\boldsymbol{\alpha }$
by the projection method  and repeat.
This procedure will converge to a limit, because the distance of the current 
vector value
of the left hand side of the equation \ref{eqn:mainfb2} 
to the right hand side is monotonically decreasing towards zero.
But the limiting distance need not be zero.
If this happens, one could use methods such as Newton-Raphson
to get out of the local minimum and repeat the earlier procedure.
\end{remark}
\subsection{An approach to the solution of the general problem}
In this subsection we make a suggestion for a practical approach
to the  problem of characterizing annihilating polynomials
achievable by output feedback, for the MIMO case.

Let  $\boldsymbol{d}(s), \boldsymbol{d}_K(s)$ be, respectively, the annihilating polynomials of $A$ and 
of $A+BKC.$  Let $\boldsymbol{b}(s)$ be the desired 
annihilating polynomial. If there is a $\boldsymbol{d}_K(s)$ 
such that 
$\boldsymbol{b}(s)=\boldsymbol{d}_K(s),$ 
we have  
$\|\boldsymbol{b}-\boldsymbol{d}_K\|=0.$ 
While $rank(K)$ might be greater 
one, it may be that there exists a matrix $K'$ of rank one 
which 
atleast satisfies $\|\boldsymbol{b}-\boldsymbol{d}_{K'}\|< \|\boldsymbol{b}-\boldsymbol{d}\|.$
If that is so, the rank one update method will help us find such a matrix.

Below, we suggest a practical approach to solving the general problem
based on  this possibility, as stated in Condition \ref{con:1}.


%
We need a  preliminary lemma which states that if we cannot move closer to 
the desired annihilating polynomial through rank one updates
using the columns of $B,$ we cannot do so even using any linear combination
of columns of $B.$ 
\begin{lemma}
Let $A,B,C$ be as in equations \ref{eqn:state11}, \ref{eqn:state22}.
Let $\omega (i)\equiv (w^0(i), \cdots, w^n(i))\ i=1, \cdots m,$
be full Krylov sequences of $A$ with their initial vectors being the columns of $B.$
Let the least square solution to 
\begin{align}
\label{eqn:mainfbi2}
\ppmatrix{ 
\rho_{1}(i)  & \cdots  &\rho_{p}(i)}C
\ppmatrix{
w^0(i)|& \cdots &|w^{n-1}(i)} 
\bbmatrix{ d_{0}  & d_{1} & \cdots &d_{n-1}\\
0&d_{0} & \cdots & d_{n-2}\\
\vdots & \vdots &\vdots & \vdots  \\
0&   \cdots & 0& d_{0}
}
\end{align}
\begin{align}
=
\ppmatrix{ 
  b_{1}& \cdots & \cdots &b_{n}}-\ppmatrix{ 
  d_{1}& \cdots & \cdots &d_{n}},
\end{align}
for $i=1, \cdots , m$ be the zero vector.
Let $$\omega \equiv 
(w^0, \cdots, w^n)= (\Sigma _i\alpha_iw^0(i)| \cdots |\Sigma _i\alpha_iw^{n}(i)).$$
Then, the least square solution to 
\begin{align}
\label{eqn:mainfbi3}
\ppmatrix{ 
\rho_{1}  & \cdots  &\rho_{p}}C
\ppmatrix{
w^0|& \cdots &|w^{n-1}} 
\bbmatrix{ d_{0}  & d_{1} & \cdots &d_{n-1}\\
0&d_{0} & \cdots & d_{n-2}\\
\vdots & \vdots &\vdots & \vdots  \\
0&   \cdots & 0& d_{0}
}
\end{align}
\begin{align}
=
\ppmatrix{ 
  b_{1}& \cdots & \cdots &b_{n}}-\ppmatrix{ 
  d_{1}& \cdots & \cdots &d_{n}},
\end{align}
is also the zero solution.
\end{lemma}
\begin{proof}
Let $ (f_{1}  ,  \cdots  ,f_{n})$ 
denote the vector on the right side of equation \ref{eqn:mainfbi2}.
When the least square solution is zero we have 
\begin{align}
\label{eqn:mainfbi4}
C\ppmatrix{
w^0(i)|& \cdots &|w^{n-1}(i)}\bbmatrix{ d_{0}  & d_{1} & \cdots &d_{n-1}\\
0&d_{0} & \cdots & d_{n-2}\\
\vdots & \vdots &\vdots & \vdots  \\
0&   \cdots & 0& d_{0}
} 
\ppmatrix{ 
 f_{1}  \\  \vdots  \\f_{n}}=
0, i=1, \cdots ,m.
\end{align}
It follows that 
\begin{align}
\label{eqn:mainfbi5}
C[\Sigma _i\alpha_i\ppmatrix{
w^0(i)|& \cdots &|w^{n-1}(i)}]\bbmatrix{ d_{0}  & d_{1} & \cdots &d_{n-1}\\
0&d_{0} & \cdots & d_{n-2}\\
\vdots & \vdots &\vdots & \vdots  \\
0&   \cdots & 0& d_{0}
} 
\ppmatrix{ 
 f_{1}  \\  \vdots  \\f_{n}}=
C\ppmatrix{
w^0|& \cdots &|w^{n-1}} 
\ppmatrix{ 
 f_{1}  \\  \vdots  \\f_{n}}=0.
\end{align}
This means the least square solution to equation \ref{eqn:mainfbi3}
is also zero. 
\end{proof}
We can now develop a practical algorithm for
MIMO output feedback case 
based on the condition that we state 
formally below.

\begin{condition}
\label{con:1}
Let $A,B,C$ be as in equations \ref{eqn:state11}, \ref{eqn:state22}.
Let $\boldsymbol{d}(s)\equiv s^n+d_1s^{n-1}+\cdots +d_n$ be the annihilating
polynomial of $A.$ Let $(b_1, \cdots , b_n)\in \Re^n.$ 
Let $A+BKC,$ for some $K,$ have the annihilating polynomial 
$s^n+b_1's^{n-1}\cdots +b_n'$ 
such that $$\|\boldsymbol{b}-\boldsymbol{b'}\|=\|(b_1, \cdots ,b_n)-(b_1', \cdots ,b_n')\|<\|(b_1, \cdots ,b_n)-(d_1, \cdots ,d_n)\|=\|\boldsymbol{b}-\boldsymbol{d}\|,$$
where $\|x\|$ denotes the $l_2$ norm of $x.$

Then, there is a matrix $K"$ with a single nonzero row such that
$A+BK"C$ has the annihilating polynomial $ s^n+b"_1s^{n-1}+\cdots +b"_n$
which satisfies 
$$\|\boldsymbol{b}-\boldsymbol{b"}\|=\|(b_1, \cdots ,b_n)-(b_1", \cdots ,b_n")\|<\|(b_1, \cdots ,b_n)-(d_1, \cdots ,d_n)\|=\|\boldsymbol{b}-\boldsymbol{d}\|,$$
\end{condition}

\begin{remark}
The above condition is stronger than required for numerical computations.
For a given coefficient vector $(b_1, \cdots , b_n)$ of the desired
annihilating polynomial, even if the condition  fails only over a measure zero set of vectors $(b_1', \cdots,b_n')$ in  $\Re^n,$ 
it would be adequate for our purposes,
since we will be generating a sequence of terms  of the kind $(b_1", \cdots,b_n")$ and testing convergence  through some
`tolerance' value for nearness of successive terms.
%
%
\end{remark}

Let $A,B,C$ be as in equations \ref{eqn:state11}, \ref{eqn:state22}.
Let $\boldsymbol{d}(s)\equiv s^n+d_1s^{n-1}+\cdots +d_n$ be the annihilating
polynomial of $A.$ 
Let $\boldsymbol{b}(s)\equiv s^n+b_1s^{n-1}+\cdots +b_n$ be the desired annihilating
polynomial for a matrix of the form $A+BK'C,$ with $(b_1, \cdots , b_n)\ne (d_1, \cdots , d_n).$

We describe below an informal algorithm which,  if Condition  \ref{con:1} is true,
will yield  a matrix $K$ with 
$A+BKC$ having its annihilating polynomial
`close' to the desired one
or with a statement that such $K$ 
can be found within the permiited number of iterations.

The output of the algorithm is based on the convergence of a sequence 
of $n-$ vectors that it generates.
\begin{algorithm}
\label{alg:1}
{\bf Input}\hspace{1cm}
Matrices $A,B,C$ as in equations \ref{eqn:state11}, \ref{eqn:state22}.\\
$\boldsymbol{d}^0(s)\equiv s^n+d_1^0s^{n-1}+\cdots +d_n^0,$ the annihilating
polynomial of $A,$\\
$\boldsymbol{b}(s)\equiv s^n+b_1s^{n-1}+\cdots +b_n,$ the desired annihilating
polynomial.\\
A positive `tolerance number' $\epsilon .$ \\
An upper bound $N$ on number of iterations.

{\bf Output} \hspace{1cm}
A matrix $K $ such that the annihilating
polynomial $\boldsymbol{d}(s)$ of $A+BKC$ satisfies 
$\|\boldsymbol{b}-\boldsymbol{d}\|< \epsilon ;$
OR\\
A statement that no such $K$ can be found for number of iterations less than $N.$\\

$\boldsymbol{d}\equiv \boldsymbol{d}^0;$ $K\equiv \mbox{zero matrix of dimension }m\times p;$\\
\noindent{While $\{$ $\|\boldsymbol{b}-\boldsymbol{d}\|\geq \epsilon  $ and number of iterations $<N \ \}$}\\
do $\{$

\hspace{1cm} For $\{ $ $j=1, \cdots m$ $\}$

\hspace{1cm} Repeat $\{$ \\

\hspace{1cm} Construct the  full Krylov sequence $(w^0, \cdots, w^n)$ of $A$ with initial vector as the  

\hspace{1cm}  $j^{th}$ column of $B.$


\hspace{1cm} As  in Subsection \ref{subsec:rank_one}, 
 compute $(\rho_{1}  , \cdots  ,\rho_{p})C (w^0| \cdots| w^n),$

\hspace{1cm} by solving equation \ref{eqn:mainfb}. 







\hspace{1cm} Compute
\begin{align} 
K_j =
 \ppmatrix{- \alpha_1\\\vdots\\-\alpha_m}
\ppmatrix{ 
 \rho_{1}  & \cdots  &\rho_{p}},
\end{align}

\hspace{1cm} where $\alpha_j=1, \ \alpha_i=0, \ i\ne j.$

\hspace{1cm} Compute $(\sigma_1, \cdots, \sigma_n) =
(\rho_{1}  , \cdots  ,\rho_{p})C (w^0, \cdots, w^n).$
Set
$\boldsymbol{\sigma}\equiv (1,\sigma_1, \cdots, \sigma_n),$ 

\hspace{1cm} Compute $\boldsymbol{d^j}\equiv \boldsymbol{\sigma}*\boldsymbol{d}.$

\hspace{1cm} Set 
$\boldsymbol{d^{j}}(s)\equiv s^n+d_1^{j}s^{n-1}+\cdots +d_n^{j}$
as the annihilating polynomial of $A+BK_jC,$

\hspace{1cm} Set $\boldsymbol{d'}\equiv \boldsymbol{d^r},$ where 
$\|\boldsymbol{b}-\boldsymbol{d^r}\| =min_{j=1}^m\|\boldsymbol{b}-\boldsymbol{d^j}\|, 
$

\hspace{1cm}
 Set $K'$ to $K_r,$ $Knew$ to  $K+K'$ and $K$ to $Knew.$

\hspace{1cm}$\}$\\


\hspace{0.3cm} reset $Anew$ to $A+BKnewC,$ $A$ to $Anew,$ $\boldsymbol{d}$ to $\boldsymbol{d'},$
$\boldsymbol{b}$ unchanged.\\

\hspace{0.1cm} $\}$

\end{algorithm}

{\it Justification based on Condition \ref{con:1}}

At any stage in the progress of the algorithm we have a current annihilating
polynomial $\boldsymbol{d}^i(s),$ a current matrix $A(i)$ and a matrix 
$K_{rem}^i,$ such that the desired polynomial $\boldsymbol{b}(s)$
annihilates $A(i)+BK_{rem}^iC.$ Therefore, by Condition \ref{con:1}, there is a matrix $K_{increm}$
with a single nonzero row such that the annihilating polynomial
$\boldsymbol{d}^{i+1}(s)$ of  $A+BK_{increm}C$ 
satisfies
$$ \|\boldsymbol{b}-\boldsymbol{d}^{i+1}\|< \|\boldsymbol{b}-\boldsymbol{d}^{i}\|.$$
The sequence $ (\cdots ,\|\boldsymbol{b}-\boldsymbol{d}^{i}\|, \cdots )$
is nonegative and monotonically decreasing. Therefore it converges
to some $\|\boldsymbol{b}-\boldsymbol{d}^{final}\|.$
We must have $\boldsymbol{b}=\boldsymbol{d}^{final}$ 
as otherwise we  can find by rank one update a $\boldsymbol{d}'$
such that $
\|\boldsymbol{b}-\boldsymbol{d}'\|< \|\boldsymbol{b}-\boldsymbol{d}^{final}\|.$\\

In the above algorithm the main computations in the inner loop
are
\begin{itemize}
\item Computation of the full Krylov sequence $(w^0, \cdots, w^n)$ of $A.$
This requires $n$ multiplications of the form $Ax.$
This is $O(n^3).$
\item Solution of the equation \ref{eqn:mainfb}.
This is $O(n^3).$
\item Computation of $K_j.$ This is  $O(n).$
\item Computation of $BK_jC.$ This is $O(mn).$
\item  Computation of $\boldsymbol{d^j}\equiv \boldsymbol{\sigma}*\boldsymbol{d}.$
This is  $O(n^2).$
(This can be speeded up by FFT but is not of consequence here.)
\end{itemize}

Let us say that an $n\times n$  matrix $T$ is available `implicitly' if, given an $n-$ vector $x$ there is a subroutine available which computes $Tx.$
We will take such a computation to be $O(n^2).$ (In practice, for instance 
in the case of large electrical networks, because of sparsity, it would be substantially less.)
Suppose $A,B,C$ are available only implicitly.
There would then be no need to compute them explicitly for the algorithm
to proceed.
The above complexity calculation would remain valid.\\


\subsection{Coefficients and roots of polynomials}
The output feedback problem, as originally posed, was about 
forcing all the eigenvalues of the system after feedback to move to the 
complex left half plane. How relevant to this problem is the 
rank one update method and Algorithm \ref{alg:1}?

The natural use of Algorithm I for shifting eigenvalues, that suggests itself, is to move the roots to a more desirable 
location altering the coefficient vector of the annihilating polynomial 
slightly. This could be done  through a rank one update using  a random
linear combination of the columns of $B$
 as the initial vector of the full Krylov sequence for $A.$ 
This would involve one inner loop iteration of Algorithm \ref{alg:1}.
More explicitly, let 
$s_1, \cdots , s_n$ be the roots of the real annihilating polynomial $\boldsymbol{d}(s).$ 
Suppose $s_1+\epsilon _1, \cdots , s_n+\epsilon _n,$
where $\epsilon _1, \cdots , \epsilon _n,$
is  a sufficiently small complex vector,
preserves complex conjugation and is in a more desirable location
(for instance, if all the roots in the right half plane move closer to the imaginary axis).
Take $\boldsymbol{b}(s)$ 
to have its roots as $s_1+\epsilon _1, \cdots , s_n+\epsilon _n$
and do the rank one update of Subsection \ref{subsec:rank_one}.\\
\subsection{Numerical evidence}
Some numerical experiments were performed to test 
\begin{itemize}
\item
the efficacy
of Algorithm I for checking if a desired annihilating polynomial
can be achieved through output feedback and 
\item 
to examine if the algorithm can be used to shift eigenvalues 
to desired locations.
\end{itemize}
The results, which are encouraging, are available in  \ref{sec:exp}.
\section{Conclusion}
\label{sec:conclusion}
We have developed an alternative way of looking at 
old ideas of state feedback that is computationally useful.

We have presented a `rank one update'  method for
modifying annihilating polynomials of state transition matrices
through output feedback. This method characterizes all 
polynomials which can be annihilating polynomials
of matrices of the kind $A+BKC,$ when $A,B,C$ are given, 
$B$ has a single column or $C$ has a single row. 

We have also presented a practical approach to handle the case where
$B$ has more than one column (MIMO) by repeatedly using the
`rank one update' method.

Numerical experiments appear to support our approach to the MIMO
case. The approach also appears promising as a way of altering
eigenvalues of systems to desired locations.

Our method alters coefficients of polynomials.
For it to work, it is necessary that the coefficient vector
of the annihilating polynomial of the matrix and its full Krylov sequence
be very nearly orthogonal.
For this to hold,
{\it methods  of  scaling coefficient vectors and of scaling roots}  will be needed
in general.
A careful error analysis taking scaling into account 
is an immediate future goal.

A harder problem is to examine if  Algorithm I can be justified in a generic sense.
\section{Acknowledgement}
Hariharan Narayanan was partially supported by a Ramanujan fellowship.

\appendix
\section{Proof of Implicit Inversion Theorem}
Below, 
when $X$, $Y$ are disjoint, a vector $f_{X\uplus  Y}$ on $X\uplus Y$ 
is written as $f_X\oplus f_Y;$
$\mathcal{K}_X$ denotes an arbitrary collections of vectors on $X;$
when $X$, $Y$ are disjoint, $\mathcal{K}_X + \mathcal{K}_Y$ is written as $\mnw{\mathcal{K}_X \oplus \mathcal{K}_Y}.$ 

The following lemma is a slightly modified version of a result (Problem 7.5) in \cite{HNarayanan1997}. Implicit Inversion Theorem is an immedate consequence.
\begin{lemma}
\label{lem:Kgenminor}
\begin{enumerate}
\item Let $\KSP,\KSQ $ be collections of vectors on $S \uplus P,S\uplus Q $  respectively.

Then there exists a collection of vectors
$\KPQ$ on $P\uplus Q$ s.t. ${\bf 0}_{PQ} \in \KPQ $ and $\KSP \lrar \KPQ = \KSQ,$ 

only if
$\KSP\circ S \supseteq \KSQ \circ S$
and
$ \KSP \times S \subseteq \KSQ \times S.$
\vspace{0.25cm}
\item Let $\KSP$ be a collection of vectors closed under subtraction  on $S \uplus P$ 

and let  $\KSQ$ be a collection of vectors on $S\uplus Q,$
closed under addition. 

Further let $ \KSP \times S \subseteq \KSQ \times S$ and $\KSP\circ S \supseteq \KSQ \circ S.$

Then the collection of vectors
$\KSP \lrar \KSQ,$ is closed under addition, with ${\bf 0}_{PQ}$ as a member 

and further we have that $\KSP \lrar (\KSP\lrar \KSQ)= \KSQ.$
\vspace{0.25cm}
\item Let $\KSP$ be a collection of vectors closed under subtraction, 
and let $\KSQ$ satisfy
the conditions, 

closure under addition, 
${\bf 0}_{SQ}\in \KSQ, \KSP \times S \subseteq \KSQ \times S$ and $\KSP\circ S \supseteq \KSQ \circ S.$

Then the equation
$$\KSP\lrar \KPQ = \KSQ,$$
where $\KPQ$ has to satisfy
closure under addition, 
has a unique solution under the condition
$$\KSP \times P \subseteq \KPQ \times P \ and\  \KSP\circ P \supseteq \KPQ \circ P.$$

If the solution $\KPQ$ does not satisfy these conditions, then there exists another solution,
i.e., $$\KSP\lrar \KSQ,$$ that satisfies these conditions.
\end{enumerate}
\end{lemma}

\begin{proof}
\begin{enumerate}
\item
Suppose $\KSP \lrar \KPQ = \KSQ$ and 
${\bf 0}_{PQ} \in \KPQ .$

It is clear
from the definition of the matched composition operation that
$\KSP \circ S \supseteq \KSQ\circ S.$  

Since 
${\bf 0}_{PQ} \in \KPQ,  $  if 
$\fS \oplus {\bf 0}_{P} \in \KSP,$ we must have that $\fS \oplus {\bf 0}_{Q} \in \KSQ.$

Thus, $\KSP \times S \subseteq \KSQ\times S.$ 
\vspace{0.5cm}
\item
On the other hand suppose
$ \KSP \times S \subseteq \KSQ \times S$ and $\KSP\circ S \supseteq \KSQ \circ S.$
\\
Let $\KPQ \equiv \KSP \lrar \KSQ,$ i.e., $\KPQ$ is the collection of all vectors $\fP\oplus \fQ $ s.t. for some
vector $\fS,$ 
$\fS\oplus \fQ \in \KSQ$, $\fS \oplus \fP \in \KSP.$  

Since $\KSP$ is closed under subtraction,
it  contains the zero vector, the negative of every vector in it and is closed
under addition. 

Since  ${\bf 0}_S \oplus{\bf 0}_P\in \KSP,$ we must have that   ${\bf 0}_S \in \KSP \times S$ and therefore   ${\bf 0}_S \in \KSQ\times S.$
It follows that 
${\bf 0}_S \oplus {\bf 0}_{Q} \in \KSQ. $ 

Hence, by definition of $\KPQ,
{\bf 0}_{PQ} \in \KPQ.$ Further, since both $\KSP, \KSQ,$ are closed under addition, so is $\KPQ$ since $\KPQ= \KSP \lrar \KSQ.$\\
Let $\fS\oplus \fQ \in \KSQ.$

Since $\KSP\circ S \supseteq \KSQ \circ S,$ for some $\fP,$ we must have that $\fS\oplus \fP \in \KSP.$
By the definition of $\KPQ,$ we have that $\fP\oplus \fQ \in \KPQ.$ 

Hence,
$\fS\oplus \fQ \in \KSP \lrar \KPQ.$ 

Thus,
$\KSP \lrar \KPQ \supseteq \KSQ.$

\vspace{0.5cm}
Next, let $\fS\oplus \fQ\in \KSP \lrar \KPQ,$ i.e., for some $\fP, \fS \oplus \fP \in \KSP$ and $\fP\oplus \fQ \in \KPQ.$

We know, by the definition of $\KPQ,$ that there exists $\fS'\oplus \fQ \in \KSQ$ s.t. $\fS' \oplus \fP \in \KSP.$

Since $\KSP$ is closed under subtraction, we must have, $(\fS - \fS') \oplus
{\bf 0}_{P} \in \KSP.$  

Hence, $\fS - \fS' \in \KSP \times S
\subseteq \KSQ\times S.$  

Hence $(\fS - \fS') \oplus
{\bf 0}_{Q} \in \KSQ.$ 

Since $\KSQ$ is closed under addition and
$\fS'\oplus \fQ \in \KSQ$,

it follows that $(\fS - \fS')\oplus {\bf 0}_{Q} + \fS'\oplus \fQ = \fS \oplus \fQ$ also
belongs to $\KSQ$.  

Thus, $\KSP \lrar \KPQ \subseteq \KSQ.$

\vspace{0.5cm}
\item From parts (1) and (2) above, the equation can be satisfied by some $\KPQ$ if and only if 
$\KSP \circ S \supseteq \KSQ\circ S$ and  $\KSP \times S \subseteq \KSQ\times S.$ 

Next, let $\KPQ$ satisfy the equation  $\KSP\lrar \KPQ =\KSQ $ and be closed under addition. 

From part (2), we know that if $\KPQ$ satisfies $\KSP \circ P \supseteq \KPQ\circ P$ and  $\KSP \times P \subseteq \KPQ\times P,$

then $\KSP\lrar (\KSP\lrar \KPQ) =\KPQ.$ But $\KPQ$ satisfies $\KSP\lrar \KPQ =\KSQ.$

It follows that for any such $\KPQ,$ we have $\KSP\lrar \KSQ=\KPQ.$ 

This proves that $\KSP\lrar \KSQ$
is the only solution to the equation $\KSP\lrar \KPQ =\KSQ, $ under the condition $\KSP \circ P \supseteq \KPQ\circ P$ and  $\KSP \times P \subseteq \KPQ\times P.$ 

\end{enumerate}
\end{proof}
\begin{remark}
We note that the  collections of vectors in Lemma \ref{lem:Kgenminor} can be over rings rather than over fields- in particular over the ring of integers. Also only $\KSP$ has to be closed over subtraction.
The other two  collections $\KPQ,\KSQ$ are  required only to be closed over addition with a zero vector as a member. In particular 
\begin{itemize}
\item $\KSP$ could be a vector space over rationals while
 $\KPQ,\KSQ$  could be cones over rationals. 
\item $\KSP$ could be a module  over integers while
 $\KPQ,\KSQ$  could be over integers and but  satisfy homogeneous linear inequality constraints.
\end{itemize}

In all these cases the proof would go through 
 without change.
\end{remark}
\section{Proof of Lemma \ref{lem:1} and Lemma \ref{lem:2}}
In the proofs  below, the definitions of intersection and sum of vector spaces
are as given in Section \ref{sec:Preliminaries}.

We need the following simple results in the proofs of the lemmas.
\begin{enumerate}
\item Let $\Vab^1 ,\Vab^2$ be vector spaces on $A\uplus B$
over $\mathbb{F}.$ Then
$$(\Vab^1\cap \Vab^2)\times A=(\Vab^1\times A)\cap (\Vab^2\times A).$$
\item Let $\Vabc $ be a vector space on $A\uplus B\uplus C.$
Then
$$(\Vabc\times AB)\circ B = (\Vabc\circ BC)\times B.$$

\end{enumerate}

\begin{proof} (Lemma \ref{lem:1})
1.
Observe that 
the solution spaces of equations \ref{eqn:state11}, \ref{eqn:state22}, are respectively   $\Vwdwuy\circ W\dw U$
and  $\Vwdwuy\circ WY.$\\
Further,
$(f_W,g_{\dw},f_U,f_Y)\in \Vwdwuy$ iff
  $(f_W,g_{\dw},f_U)$ belongs to the 
solution space of equation \ref{eqn:state11} and \\
$(f_W,f_Y)$
belongs to the solution space of equation \ref{eqn:state11}.\\
Equivalently, $\Vwdwuy=(\Vwdwuy\circ W\dw U)\cap (\Vwdwuy\circ WY).$\\
Since  $w$ is free in
the solution space of the
equation \ref{eqn:state22} and rows of $C$ are linearly independent
this means that $y$ is free in the solution space of equation \ref{eqn:state22}, i.e., in $\Vwdwuy\circ WY.$\\
Let $y=f_Y.$
We will show that for some 
$f_U,$ $(f_U,f_Y)\in (\Vwdwuy \cap \Vwdw^2)\circ UY.$
It will follow that $y$ is free in $(\Vwdwuy \cap \Vwdw^2)\circ UY.$\\
For some $f_W$
we have that  $(f_W,f_Y)\in  \Vwdwuy\circ WY  = \mbox{solution space of equation \ref{eqn:state22}} .$\\
Since $w$ is free in $\Vwdw^2,$
there exists  some $g_{\dw}$ such that 
$(f_W,g_{\dw})\in \Vwdw^2\subseteq \Vwdwuy\circ W\dw.$\\
Therefore there exists some $f_U$ such that $(f_W,g_{\dw},f_U) \in \Vwdwuy\circ W\dw U .$
\\
Now $(f_W,f_Y)\in \Vwdwuy\circ WY$ and 
$(f_W,g_{\dw},f_U) \in \Vwdwuy\circ W\dw U .$
Therefore,
$(f_W,g_{\dw},f_U,f_Y)\in \Vwdwuy.$\\ 
Since $(f_W,g_{\dw})\in \Vwdw^2,$ it follows that $(f_W,g_{\dw},f_U,f_Y)\in\Vwdw^2\oplus \F_U\oplus \F_Y.$\\
Thus $(f_W,g_{\dw},f_U,f_Y)\in \Vwdwuy \cap \Vwdw^2$ and therefore
$(f_U,f_Y)\in (\Vwdwuy \cap \Vwdw^2)\circ UY.$

2. By definition, $\Vwdwuy\leftrightarrow \Vwdw^2$ is 
$(\Vwdwuy\cap \Vwdw^2)\circ UY.$ Since $y$ is free in $(\Vwdwuy \cap  \Vwdw^2)\circ UY,$ it is also free in $(\Vwdwuy \leftrightarrow  \Vwdw^2)\circ UY\circ Y,$
i.e., 
 $(\Vwdwuy\leftrightarrow \Vwdw^2)\circ Y=\F_Y.$
\end{proof}

\begin{proof} ((Lemma \ref{lem:2})
1. $\Vwdwuy\times \wdw $ is the  solution space of $\dsw=Aw,0=Cw.$ Since this is contained in
the solution space of $\dsw=\hat{A}w,$ we must have that the solution space of
$\dsw=Aw,0=Cw,$ is the same as the solution space of $\dsw=Aw,0=Cw,\dsw=\hat{A}w.$
Thus $Aw= \hat{A}w$ when $Cw=0.$ Suppose $y$ is set to zero in the solution space of
Equations \ref{eqn:state11}, \ref{eqn:state22}.
We then have $Cw=0$ so that $\dsw=Aw=\hat{A}w.$
But $\dsw=Aw+Bu$ and columns of $B$ are linearly independent.
We conclude that $u=0$ as required.

Thus if $(f_W,f_{\dw},f_U,0_Y)\in \Vwdwuy,$ we must have that
$f_U=0_U$ and therefore that $(f_W,f_{\dw})\in \Vwdwuy\times W\dw.$
It follows that  $\Vwdwuy\times W\dw U= \Vwdwuy\times W\dw\oplus 0_U.$

2. We have $$(\Vwdwuy\leftrightarrow \Vwdw^2)\times U=((\Vwdwuy\cap \Vwdw ^2)\circ UY)\times U
= ((\Vwdwuy\cap \Vwdw ^2)\times W\dw U)\circ U.$$
Now $\Vwdwuy\cap \Vwdw ^2=\Vwdwuy\cap (\Vwdw ^2\oplus \F_U\oplus \F_Y)$
so that $(\Vwdwuy\cap \Vwdw ^2)\times W\dw U$
$$=(\Vwdwuy\cap (\Vwdw ^2\oplus \F_U\oplus \F_Y))\times W\dw U= (\Vwdwuy\times W\dw U)\cap ((\Vwdw ^2\oplus \F_U\oplus \F_Y)\times W\dw U)).$$
By the previous part of this lemma, 
$\Vwdwuy\times W\dw U=\Vwdwuy\times W\dw\oplus 0_U$
and by the hypothesis of this lemma
$\Vwdwuy\times W\dw\subseteq \Vwdw ^2.$
Therefore, $\Vwdwuy\times W\dw U\subseteq (\Vwdw ^2\oplus \F_U\oplus \F_Y)\times W\dw U.$  
Hence, $(\Vwdwuy\cap \Vwdw ^2)\times W\dw U=\Vwdwuy\times W\dw\oplus 0_U.$
Therefore, $(\Vwdwuy\leftrightarrow \Vwdw^2)\times U$ 
$$=((\Vwdwuy\cap \Vwdw ^2)\circ UY)\times U= ((\Vwdwuy\cap \Vwdw ^2)\times W\dw U)\circ U
=(\Vwdwuy\times W\dw\oplus 0_U)\circ U=0_U.$$

\end{proof}
\section{Numerical experiments}
\label{sec:exp}
A simpler version of Algorithm I would be  to replace the inner loop of the algorithm, where rank one update is attempted
with each column of $B,$  by random linear combinations of the columns of $B.$
This simpler version  was coded in the open source computational package Scilab.

Two different experiments were performed.

The purpose of the first experiment was to check whether, given $A,B,C$ and 
an annihilating polynomial $\boldsymbol{b}(s)$
 of  matrix $A+BKC,$ 
by using Algorithm I, but without knowledge of $K,$ we can find a  matrix\\ $A+BK_{final}C,$ whose 
 annihilating polynomial  $\boldsymbol{d}_{final}(s)$
is close to $\boldsymbol{b}(s).$ If the  polynomials
are close, it would support the claim that MIMO case
can be solved by repeated rank one updates.

An additional purpose was to check if we can  clearly distinguish the above from the case where annihilating polynomial $\boldsymbol{b}(s)$
does not correspond to a matrix of the form $A+BKC.$

The purpose of the second experiment was to examine, whether the algorithm
can be practically used to alter eigenvalues using only output feedback.

\subsection{First experiment}
The experiment was performed with $n=20,m=3,p=3.$ 
It is known that for $m\times p > n,$ output feedback is 
generically feasible to put all eigenvalues in left half of the complex plane.
Hence lower values of $m$ and $p,$ were chosen.
The polynomial $\boldsymbol{d}(s)\equiv d_0s^n+\cdots +d_n$ 
was generated randomly with integer coefficients
between $-2$ and $+2,$ with $d_0=1 $ and $d_n$ randomly as either $+1$ or $-1.$
The matrix $A$ was taken to be in the companion form with
its annihilating polynomial as $\boldsymbol{d}(s).$
Matrices $B,C$ were picked as integer matrices with entries chosen randomly
between $-10$ and $+10.$
The matrix $K$ was chosen with entries randomly between $-0.001$ and $0.001.$ 


Fifty repetitions of the first experiment, as described below, were performed using Algorithm I. Each repetition involved a different randomly chosen
quadruple of matrices $A,B,C,K.$\\ 
The annihilating polynomial 
$\boldsymbol{b}(s)$ was 
computed for $A+BKC.$ (Thus we know that $\boldsymbol{b}(s)$ is reachable.)
The desired polynomial was fixed as $\boldsymbol{b}(s).$
The polynomial $\boldsymbol{d}(s)$ (annihilating polynomial of $A$) was updated, 
through repeated rank one updates corresponding to random linear combinations of $col(B),$
to $\boldsymbol{d}_{final}(s)$
which corresponded to the updated matrix $A+BK_{final}C.$
The  rank one update was performed for $1000$ iterations and for
$1500$ iterations reaching respectively $\boldsymbol{d}_{1000}(s)$
and $\boldsymbol{d}_{1500}(s),$

Table $\frac{\|\boldsymbol{d}_{final}-\boldsymbol{b}\|}{\|\boldsymbol{b}\|}$
indicates how many of the repetitions gave values of $\frac{\|\boldsymbol{d}_{final}-\boldsymbol{b}\|}{\|\boldsymbol{b}\|}$ within the intervals
$10^{-15} \ to \ 10^{-14},
\cdots , 10^{-8} \ to \ 10^{-7}$ after $1000$ and $1500$ rank one update iterations. 

It is clear that we are able to reach quite close to the desired annihilating polynomial in all the repetitions.
\\

Table $\frac{\|(A+BK_{final}C)-(A+BKC)\|}{\|A\|}$
indicates how many of the repetitions gave values of 
$\frac{\|(A+BK_{final}C)-(A+BKC)\||}{\|A\|}$
within the intervals
$10^{-15} \ to \ 10^{-14},
\cdots , 10^{-8} \ to \ 10^{-7}$ after $1000$ and $1500$ rank one update iterations.

Note that $K$ was originally generated randomly.
The rank one update method keeps adding to the $m\times p$ zero matrix
a series of rank one matrices. It is  surprising that we actually reach very close to the
unknown $K,$
since there is no reason to believe that 
the annihilating polynomial uniquely fixes $K$ in the MIMO case.

%
%
%
%
%
%

\begin{table}
\parbox{.45\linewidth}{
\centering
\begin{tabular}{ccc}
\hline\hline

Interval &  $1000$ iter & $1500$ iter  \\ [0.5ex]
\hline
$10^{-15}\ to\  10^{-14}$&$1$ &$1$\\
$10^{-14}\ to\  10^{-13}$&$2$ &$5$\\
$10^{-13}\ to\  10^{-12}$&$12$ &$12$\\
$10^{-12}\ to \ 10^{-11}$&$15$ &$18$\\
$10^{-11}\ to \ 10^{-10}$&$12$ &$9$\\
$10^{-10}\ to \ 10^{-9}$&$7$ &$5$\\ 
$10^{-9}\ to \ 10^{-8}$&$0$ &$0$\\ 
$10^{-8}\ to \ 10^{-7}$&$1$ &$0$\\ [1ex]
\hline
\end{tabular}
\caption{$\frac{\|\boldsymbol{d}_{final}-\boldsymbol{b}\|}{\|\boldsymbol{b}\|}$ }
}
\hfill
\label{tab:2}
\parbox{.45\linewidth}{
\centering
\begin{tabular}{ccc}
\hline\hline

Interval &  $1000$ iter & $1500$ iter  \\ [0.5ex]
\hline
$10^{-15}\ to\  10^{-14}$&$1$ &$1$\\
$10^{-14}\ to\  10^{-13}$&$0$ &$0$\\
$10^{-13}\ to\  10^{-12}$&$7$ &$8$\\
$10^{-12}\ to \ 10^{-11}$&$15$ &$16$\\
$10^{-11}\ to \ 10^{-10}$&$13$ &$14$\\
$10^{-10}\ to \ 10^{-9}$&$11$ &$10$\\
$10^{-9}\ to \ 10^{-8}$&$2$ &$1$\\
$10^{-8}\ to \ 10^{-7}$&$1$ &$0$\\ [1ex]
\hline
\end{tabular}
\caption{$\frac{\|(A+BK_{final}C)-(A+BKC)\|}{\|A\|}$ }
}
\end{table}

{\it What happens when the desired annihilating polynomial 
cannot be reached?}

It is known that for $m\times p <n,$  the output feedback
problem is generically infeasible.

Since we chose $n,m,p$ such that $m\times p <n,$ a randomly chosen
$\boldsymbol{b}$ will correspond to an unreachable annihilating polynomial.

Fifty repetitions of the following experiment were performed.

The annihilating polynomial $\boldsymbol{d}(s)\equiv d_0s^n+\cdots +d_n$
was generated randomly with integer coefficients
between $-2$ and $+2,$ with $d_0=1 $ and $d_n$ randomly as either $+1$ or $-1.$
The desired annihilating polynomial
$\boldsymbol{b}(s)$ was obtained by adding a polynomial $\boldsymbol{\epsilon}(s)$
to $\boldsymbol{d}(s).$ The coefficients of this latter polynomial
were picked at random between $-0.01$ and $+0.01.$
Further, $\epsilon(0),\epsilon(n)$ were set to zero.
The matrix $A$ was taken to be in the companion form with
its annihilating polynomial as $\boldsymbol{d}(s).$ 
Matrices $B,C$ were picked as integer matrices with entries chosen randomly
between $-10$ and $+10.$
The rank one update was performed for $1000$ iterations.

The rank one update method yielded vectors $\boldsymbol{d}_{newj}$
such that $\frac{\|\boldsymbol{d}_{newj}-\boldsymbol{b}\|}{\|\boldsymbol{b}\|}$ 
converged to a value
greater than $\delta >0.02.$ 
However even while this convergence takes place, $\boldsymbol{d}_{newj}$ cycled through a set of values (each of which can be taken as reachable).

Next, the desired polynomial was kept as a polynomial $\boldsymbol{d}_{new}(s)$
which was reached through the above sequence of rank one updates
 while starting with matrix $A$ with its  annihilating polynomial
$\boldsymbol{d}(s).$
The rank one update was performed for $1000$ iterations 
reaching  $\boldsymbol{d}_{1000}(s)$
as the annihilating polynomial of the updated matrix $A+BK_{1000}C.$

The 
minimum and maximum values of 
$\frac{\|\boldsymbol{d}_{1000}-\boldsymbol{d}_{new}\|}{\|\boldsymbol{d}_{new}\|}$
 over 
the $50$ repetitions are given
below.
$$\frac{\|\boldsymbol{d}_{1000}-\boldsymbol{d}_{new}\|}{\|\boldsymbol{d}_{new}\|}(min,max)=(6.737\times 10^{-22}, 1.057\times 10^{-12}
)$$

Note that in the randomly chosen $\boldsymbol{b}(s)$ case
with $1000$ iterations 
$$\frac{\|\boldsymbol{d}_{final}-\boldsymbol{d}_{desired}\|}{\|\boldsymbol{d}_{desired}\|}\equiv \frac{\|\boldsymbol{d}_{newj}-\boldsymbol{b}\|}{\|\boldsymbol{b}\|}>0.02,$$\\
whereas in the reachable $\boldsymbol{d}_{new}(s)$ case with $1000$ iterations $$\frac{\|\boldsymbol{d}_{final}-\boldsymbol{d}_{desired}\|}{\|\boldsymbol{d}_{desired}\|}\equiv \frac{\|\boldsymbol{d}_{1000}-\boldsymbol{d}_{new}\|}{\|\boldsymbol{d}_{new}\|}<10^{-11}.
$$
 
We concluded that the rank one update method can indeed distinguish 
between reachable and unreachable annihilating polynomials.

\subsection{Second experiment}
Suppose $A,B,C$ are such that for some unknown $K,$
the eigenvalues of $A+BKC$ lie in the left half of the complex plane.
Can Algorithm I help us find such a $K?$

For conducting the experiment, we first needed to build a feasible triple $A,B,C$ for which such $K$
can be found. For this, we began with $A'$ whose eigenvalues were in the left half plane
and picked random integral $B,C$ as in the first experiment.
We picked the desired annihilating polynomial to have some roots in the right half
complex plane.
We used Algorithm I and increased the number of iterations 
until the annihilating polynomial of the current matrix $A'+BK'C,$ had some roots in the right half
complex plane also. 

The resulting $A'+BKC$ was taken as the matrix $A$
and $B,C$ were left unchanged. The experiments were performed on this triple $A,B,C,$
with $\boldsymbol{d}(s)$ denoting  the annihilating  polynomial of $A.$

We worked with $n=10, m=2,p=3.$ 

Let $\boldsymbol{d}_{new}(s)$  be the current annihilating polynomial.
Initially $\boldsymbol{d}_{new}(s)\equiv\boldsymbol{d}(s).$
Let $droot_{new}$ be the roots of $\boldsymbol{d}_{new}(s).$
We built $broot_{new}$ 
by altering $droot_{new}$ using the following rule:
 
 if $real(droot_{new}(i))>0$  then $$ broot_{new}(i)=droot_{new}(i)-a*abs(droot_{new}(i))-b*(norm(droot_{new}))$$
$$\mbox{else} \ \ broot_{new}(i)=droot_{new}(i)$$
and also the same rule but with the `if $real(droot_{new}(i))>0$ '
condition omitted.

We then took the current desired annihilating polynomial to be
the polynomial $\boldsymbol{b}_{new}(s)$ whose roots were the $broot_{new}(i)
, \ i=1, \cdots n,$
and did a rank one update of $\boldsymbol{d}_{new}(s).$

Algorithm I was performed with  number of iterations between $10$
and $200.$
The values of $a,b$ were varied respectively between $.1$ to $.0001$ 
and $.001$ to $.000001$ as the iterations proceeded,
depending on how close the current values of $droot_{new}(i)$ were to 
the imaginary axis when the real part was positive.

We worked with $10$ random instances of the problem. In each case 
we could reach a $\boldsymbol{d}_{new}(s)$ with roots entirely in the 
left half plane, 
after a total of
less than $200$ iterations.
 
\bibliographystyle{elsarticle-num}
\bibliography{references}

\end{document}